\documentclass{CSML}
\pdfoutput=1

\usepackage{lastpage}
\newcommand{\dOi}{13(4:10)2017}
\lmcsheading{}{1--\pageref{LastPage}}{}{}%
{Apr.~07,~2017}{Apr.~13, 2018}{}

\usepackage{hyperref}
\hypersetup{hidelinks}

\keywords{inductive theorem proving, arithmetical theories, proof theory}

\usepackage[utf8]{inputenc}
\usepackage{amsmath}
\usepackage{amssymb}
\usepackage{mathrsfs}
\usepackage{tikz}
\usepackage{proof}
\usepackage{url}

\input{macros}

\begin{document}

\title{Some observations on the logical foundations of inductive theorem proving}

\author[S.~Hetzl]{Stefan Hetzl}
\address{
 Institute of Discrete Mathematics and Geometry\\
 Vienna University of Technology\\
 Wiedner Hauptstra\ss e~8--10\\
 AT-1040 Vienna\\
 Austria
}
\email{stefan.hetzl@tuwien.ac.at}
\thanks{Supported by the Vienna Science and Technology Fund (WWTF) project VRG12-004}

\author[T.~L.~Wong]{Tin Lok Wong}
\address{
 Kurt G\"odel Research Center for Mathematical Logic\\
 University of Vienna\\
 W\"ahringer Stra\ss e~25\\
 AT-1090 Vienna\\
 Austria
}
\email{tin.lok.wong@univie.ac.at}
\thanks{Supported by Austrian Science Fund (FWF) project P24654-N25}

\begin{abstract}
In this paper we study the logical foundations of automated inductive theorem proving. To that aim we
first develop a theoretical model that is centered around the difficulty
of finding induction axioms which are sufficient for proving a goal.

Based on this model, we then analyze the following aspects: the choice of a proof shape,
the choice of an induction rule and the language of the induction formula. In particular, using model-theoretic
techniques, we clarify the relationship between notions of inductiveness that have been considered in the
literature on automated inductive theorem proving.

This is a corrected version of the paper \url{https://arxiv.org/abs/1704.01930v5} published originally on Nov.~16, 2017.
\end{abstract}

\maketitle

\section{Introduction}
Theories of (natural number) arithmetic have been of great interest
 to mathematical logicians since the beginning of the 20th century.
At that time, work was mainly devoted to questions of consistency.
Deep connections that link
 the definability-theoretic aspects of arithmetic and
 theories of computation
  were later discovered and fruitfully developed.
In the 1970s and 80s, relevant model-theoretic techniques
 became mature enough for
  establishing mathematically interesting unprovability results.
While the model theory of arithmetic evolved into a subject of its own,
 the connections with theories of computation found their way
  down to the complexity-theoretic level.
Nowadays, theories of arithmetic
 have penetrated almost every branch of mathematical logic,
  including mathematical philosophy.

Rather independently of this work, the subject of inductive theorem proving
developed in computer science. In this tradition, the central aim is
to develop algorithms that find proofs by induction and to implement these algorithms
efficiently.
This subject is characterized by a great variety of different methods (and systems
implementing these methods), for example, rippling~\cite{Bundy05Rippling}, theory
exploration~\cite{Claessen13Automating}, integration
into a superposition prover~\cite{Kersani13Combining,Cruanes15Extending,WandXXAutomatic},
recursion analysis~\cite{Boyer79Computational,Stevens88Rational,Bundy89Rational}, proof by consistency~\cite{Comon01Inductionless},
and cyclic proofs~\cite{Brotherston11Sequent,Brotherston12Generic}.
Recently, a benchmark suite for inductive theorem proving
has been presented~\cite{Claessen15TIP}.

The aim of our work is to apply methods and results from the former tradition in mathematical logic
to the tradition in computer science.
The main advantage of this combination is that it is possible to obtain {\em unprovability} results
(by model-theoretic means) where previously in the literature on inductive theorem proving
only empirical observations could be made based on the failure of {\em a specific algorithm} to find a proof.

A first obstacle in realizing such an application is that the above mentioned approaches to inductive
theorem proving are quite different. This makes it difficult to provide a common theoretical basis.
However, the final result
is typically, in one way or another, explicitly or implicitly, a proof of
the goal from instances of an induction scheme and basic axioms from a background theory. We take this
observation as a guiding principle for the development of a theoretical model of inductive theorem
proving in Section~\ref{sec:model}. This is the main conceptual contribution of this paper.

In terms of technical contributions we analyze the following aspects of methods for inductive theorem
proving:
(a) the choice of a proof shape,
(b) the choice of the induction rule, and
(c) the language of the induction formula.
Despite the differences between the existing methods for inductive theorem proving, these aspects play a role
in most of them. Mathematically, the technical contributions of this paper are:
we show that the equivalence proof shape for inductive proofs (see Section~\ref{sec:proof_shapes}) is complete
but that the uniform proof shape is not. We show that Walther's method for comparing induction
axioms~\cite{inproc:computindax} is not complete (Section~\ref{sec:walther}). We make the (possibly surprising)
observation that the weakest induction axiom for proving an induction axiom may not be that induction axiom itself (Proposition~\ref{prop:weaker-indn}).
We establish the strictness of the implications between several frequently used notions of
inductiveness (Section~\ref{sec:comparing_notions}). We also include the result, due to Kaye, that $\PAminus$
proves the least number principle in the language of rings (Theorem~\ref{thm:PAminus-IOpen_r}).

This paper is structured as follows: after developing our model of inductive theorem proving in Section~\ref{sec:model},
we study the completeness of proof shapes in Section~\ref{sec:proof_shapes}. In Section~\ref{sec:forms_of_induction}
we describe different formulations of induction and study their equivalence in both a general and a quantifier-free context.
In Section~\ref{sec:non-closure_of_cuts} we establish non-closure properties of $\PAminus$-cuts that will
be used in the rest of the paper. In Section~\ref{sec:comparing_solutions} we investigate
two ways of comparing different induction formulas that prove the same theorem. In Section~\ref{sec:comparing_notions}, which is central to
this paper, we compare different notions of inductiveness. (A formula is called inductive
in the sense of a particular induction rule if it satisfies the hypotheses of this rule provably in the base theory.) In Section~\ref{sec:order} we study
the effect of the choice of the language for the induction formula using an example involving the $<$-relation.

This is a corrected version of the paper \url{https://arxiv.org/abs/1704.01930v5} published originally on Nov.~16, 2017.
As Emil Je\v{r}\'{a}bek pointed out, two questions posed in the original version already had answers in the literature,
 and Shepherdson almost proved the main result of Section~\ref{sec:order}.
The answers to these questions are incorporated in this paper as Theorem~\ref{thm.johannsen} and Theorem~\ref{thm.IAtomic=IOpen},
 and we added attribution to Shepherdson's work in Section~\ref{sec:order}.

\section{A theoretical model of inductive theorem proving}\label{sec:model}

This section is devoted to developing a theoretical model of inductive theorem proving. This is
necessary in order to provide a conceptually and formally clear basis for relating mathematical results
to algorithms in automated deduction.

\subsection{First-order theories of the natural numbers}

In this paper we restrict our attention to first-order theories of the natural numbers (as opposed
to, e.g., general inductive data types which is a more common choice in inductive theorem proving). This is
a pragmatic choice which is motivated by the following reasons: (1)~the mathematical study of these theories
is very well developed, and (2)~the central problems of inductive theorem proving also surface in
this restricted setting. To the mathematical logician this may not seem a restriction since coding
allows us to describe an arbitrary inductive data type as a set of natural numbers. However, the increase of the syntactic
complexity of formulas and proofs based on coding renders this approach unfit for practical applications,
 where bounded quantifiers are usually as costly as unbounded quantifiers.

Throughout this paper, we work over the base theory~$\PAminus$.
It is finitely axiomatized, induction-free, and a fragment of Peano arithmetic ($\PA$).
Robinson's~$Q$ is another commonly used base theory for arithmetic.
There are two main reasons for choosing $\PAminus$ instead of~$Q$
 as the base theory in this paper.
First, if one adopts~$Q$ as the base theory,
 then even a slight change in the definitions
  can make a significant difference in the results.
This is due to the fact that many simple arithmetic properties,
 for example,
 \begin{equation*}
  \fa{x,y}{(x<y+1\nsc x\leq y)},
 \end{equation*} are not provable in~$Q$;
 cf.~Lemma~\ref{lem:PA-}.
We do not want such details to distract us.
Second, we can extract useful information about notions of inductiveness
 in Section~\ref{sec:comparing_notions} over~$\PAminus$.
On the contrary, Robinson's~$Q$ is too weak
 to prove any non-trivial implication between our notions of inductiveness.
More specifically,
 Proposition~\ref{prop:indrel}\partref{prop:indrel:n+1_implies_<}
                          and~\partref{part:indrel/ind'=wind}
  become false if one changes the base theory to~$Q$.
Nevertheless, many other results in this paper, for example,
  Theorem~\ref{thm:Th(PA)=indn}, Proposition~\ref{prop:unif-ncompl}
  and Lemma~\ref{lem:+nx},
 remain true over~$Q$.
\begin{defi}
Denote the language $\{0,1,{+},{\times},{<}\}$ for ordered rings
 by \defm{$\Lor$}.
Abusing notation, if $n\in\IN$, then we denote the closed $\Lor$~term
 \begin{equation*}
  (\cdots((0+\underbrace{1)+1)+\dots+1)}_{\text{$n$-many $1$'s}}
 \end{equation*} by~$n$.
Let \defm{$\PAminus$} denote
 the theory of the non-negative parts of discretely ordered rings.
We axiomatize~$\PAminus$ by the following.
\begin{enumerate}[leftmargin=9mm]
  \renewcommand{\theenumi}{P\arabic{enumi}}
\item $\fa x{\fa y{\fa z{\bigl((x+y)+z=x+(y+z)\bigr)}}}$.
\item $\fa x{\fa y{(x+y=y+x)}}$.
\item $\fa x{\fa y{\fa z{\bigl((x\times y)\times z=x\times (y\times z)\bigr)}}}$.
\item $\fa x{\fa y{(x\times y=y\times x)}}$.
\item $\fa x{\fa y{\fa z{(x\times(y+z)=(x\times y)+(x\times z))}}}$.\label{ax:PA-/distrib}
\item $\fa x{(x+0=x)}$.\label{ax:PA-/+0}
\item $\fa x{(x\times0=0)}$.
\item $\fa x{(x\times1=x)}$.\label{ax:PA-/x1}
\item $\fa x{\fa y{\fa z{(x<y\wedge y<z\then x<z)}}}$.\label{ax:PA-/<trans}
\item $\fa x{\neg x<x}$.\label{ax:PA-/<irref}
\item $\fa x{\fa y{(x<y\vee x=y\vee x>y)}}$.\label{ax:PA-/<linear}
\item $\fa x{\fa y{\fa z{(x<y\then x+z<y+z)}}}$.\label{ax:PA-/<+}
\item $\fa x{\fa y{\fa z{(z\not=0\wedge x<y\then x\times z<y\times z)}}}$.
\item $\fa x{\fa y{(x<y\nsc\ex z{((x+z)+1=y)})}}$.\label{ax:PA-/<}
\item $0<1\wedge\fa x{(x>0\then x\geq1)}$.\label{ax:PA-/discrete}
\item $\fa x{(x\geq0)}$.\label{ax:PA-/nneg}
\end{enumerate}
Here \defm{$x\leq y$} is an abbreviation for $x<y\vee x=y$.
\end{defi}

\begin{defi}
If $\theta(x,\bar z)$ is an $\Lor$~formula,
 then the \defm{induction axiom for~$\theta$} with respect to the variable~$x$,
   denoted by \defm{$\ind_x\theta$} or simply \defm{$\ind\theta$},
  is the sentence
   \begin{equation*}
    \fa{\bar z}{\bigl(
     \theta(0,\bar z)
     \wedge\fa x{\bigl(\theta(x,\bar z)\then\theta(x+1,\bar z)\bigr)}
     \then\fa x{\theta(x,\bar z)}
    \bigr)}.
   \end{equation*}
Define
\begin{align*}
\IOpen & =\PAminus+\{\ind_x\theta:
 \text{$\theta(x,\bar z)$ is a quantifier-free $\Lor$~formula}
\},\\
\ISigma_k & = \PAminus+\{\ind_x\theta:
  \text{$\theta(x,\bar z)$ is an $\Lor$~formula of complexity $\Sigma_k$}
\},\ \text{and}\\
\PA & = \PAminus+\{\ind_x\theta:
  \text{$\theta(x,\bar z)$ is an $\Lor$~formula}
\}.
\end{align*}
\end{defi}

\subsection{The necessity of non-analyticity}\label{sec.nec_non_an}

It has often been observed in the literature on inductive theorem proving
that ``within first-order theories that include induction, the cut rule
cannot be eliminated''~\cite[p.~55]{Bundy05Rippling}. This observation
will provide a crucial foundation for the justification of our model
of inductive theorem proving. Therefore we would like to discuss it here
and, in the process, make it more precise and show how to prove it.

Since we want to speak about cut-elimination we need to speak about the sequent calculus. Which variant of the sequent calculus we use is not of importance for the points discussed here; for the sake of precision let us
fix it to be the calculus {\LK} of~\cite{Buss98Introduction}. A sequent is denoted as $\Gamma \seq \Delta$.
For a theory $T$ and a formula $\varphi$, $T\proves \varphi$ if and only if there is a finite
set $T_0\subseteq T$ and an {\LK}-proof of the sequent $T_0 \seq \varphi$.
Gentzen's cut-elimination theorem states that:

\begin{thm}
If there is an {\LK}-proof of a sequent $\Gamma \seq \Delta$, then there is a cut-free {\LK}-proof of $\Gamma \seq \Delta$.
\end{thm}

An important feature of cut-free proofs is that they have the subformula property. In the context
of first-order logic this means that every formula that occurs in a cut-free proof of the sequent
$\Gamma\seq\Delta$ is an \emph{instance} of a subformula of a formula that occurs in $\Gamma\seq\Delta$. A
proof that has the subformula property is also called {\em analytic}.

Since the cut-elimination theorem considers arbitrary first-order sequents, it can also be applied to
theories containing induction axioms:

\begin{cor}
If $\PA\proves \varphi$ then there is a finite $A_0\subseteq \PA$ and a cut-free {\LK}-proof of the sequent
$A_0\seq \varphi$.
\end{cor}

So we see that {\em in the sense of the above corollary}, inductive theories do allow cut-elimination.
{\em However}, $A_0$ may contain induction axioms on induction formulas which are not instances of
subformulas of $\varphi$, i.e., non-analytic induction formulas.

The {\em necessity} of non-analytic induction formulas follows, for example, from G\"odel's second
incompleteness theorem: recall that,
by arithmetizing the syntax of formulas and proofs, one can formulate the consistency of a recursive arithmetical
theory as an $\Lor$ sentence. More specifically, for all $k\in\IN$ there is a $\Pi_1$~sentence $\Con(\ISigma_k)$ expressing the consistency of $\ISigma_k$,
see for example~\cite{Buss98First}. We then have:

\begin{thm}\label{thm.on.ISigmak}
For all $k\in\IN$: $\PA\proves \Con(\ISigma_k)$ but $\ISigma_k \nproves \Con(\ISigma_k)$.
\end{thm}

Note that this result embodies a very strong non-analyticity requirement: given any $k\geq 1$,
in order to prove $\Con(\ISigma_k)$ not only do we need a non-analytic induction formula,
but we need one with more than $k$ quantifier alternations even though $\Con(\ISigma_k)$ is only a $\Pi_1$~sentence.

This theorem entails the necessity of cut in the following sense. First, formulate induction as the inference rule
\[
\infer[\mathrm{Ind}]{\Gamma\seq\Delta, \forall x\, \psi(x)}{
  \Gamma\seq\Delta,\psi(0)
  &
  \Gamma, \psi(x) \seq \Delta, \psi(s(x))
}
\]
with the usual side condition and $\psi$ being an arbitrary formula. Observe that $\PA\proves \varphi$ if and only if
there is an $\LK+\Ind$-proof of $\PAminus\seq \varphi$.
Now, in contrast to $\LK$, the calculus $\LK+\Ind$ does not have cut-elimination:

\begin{cor}
There is a formula $\varphi$ such that $\PAminus\seq \varphi$ has an $\LK +\Ind$-proof but no cut-free $\LK+\Ind$-proof.
\end{cor}
\begin{proof}
Let $\varphi = \Con(\ISigma_k)$ for {\em any} $k\geq 3$. Then, by Theorem~\ref{thm.on.ISigmak},
$\PA\proves \Con(\ISigma_k)$ and consequently there is an
$\LK+\Ind$-proof of $\PAminus \seq \Con(\ISigma_k)$. On the other hand, suppose there would be a cut-free
$\LK+\Ind$-proof of $\PAminus \seq \Con(\ISigma_k)$.
Then, due to the subformula property, all formulas, and in particular: all induction formulas, in this proof
would be $\Sigma_3$ thus contradicting Theorem~\ref{thm.on.ISigmak}.
\end{proof}

These considerations show that the observation formulated at the beginning of
this section can be stated more precisely, and without mentioning the cut rule,
as: {\em inductive theorem proving requires the use of non-analytic induction axioms}.

Remember that the notion of analyticity is to be understood in the sense of first-order logic here.
If we would move to second- or even higher-order logic, for
 example by formulating the theorem-proving problem in terms of the second-order induction axiom,
 then, in the presence of second-order quantifiers, the meaning of the subformula property, and with it that
 of analyticity, changes.
Indeed, the non-analytic induction axioms of a $\PA$\nobreakdash-proof translate to instances (in the sense of second-order logic)
 of the second-order induction axiom.
Thus a $\PA$-proof translates into a proof which is analytic in the sense of second-order logic.
But this is merely a change in terminology, not in substance and therefore we will not consider this option.

\subsection{A computational observation}\label{sec:compass}

In this section we take a closer look at the practice of inductive theorem proving and, in particular,
 at the aspect of non-analytic induction axioms as described above.

The set of sentences which are realistic as input to an inductive theorem prover in practice is naturally
 fuzzy and we cannot claim to make it precise here.
What we have in mind are goals such as those of the
 TIP library~\cite{Claessen15TIP}.
They typically consist of a universally quantified atomic formula to be proved from some background
 axioms consisting also of universally quantified atomic formulas.
They have a size between a few dozen to several hundred symbols and proofs with a symbolic complexity
 of one or two orders of magnitude that of the goal.
In comparison to arbitrary arithmetical sentences, this is a quite restricted class.
Its consideration is justified on the one hand by the practical interest in goals of this form and, on the other hand,
 by the difficulty of the inductive theorem proving problem in general.
We also restrict our attention to input sentences $\sigma$ such that $\PA\proves \sigma$.
 The recognition of non-theorems as such, while clearly of high practical value, is a different topic
 which we do not consider in this paper.

Now, given a realistic input sentence $\sigma$, there are, in theory, the following options
 for generating (non-analytic) induction axioms which  suffice to prove $\sigma$.
(1)~One can use coding, for example as in the proof of the finite axiomatizability of $\ISigma_k$
 in~\cite[Theorem~I.2.52]{book:hajek+pudlak}.
Thus one can formulate the theorem-proving problem in $\ISigma_k$ (for a $k$ sufficiently high for practical
 purposes) as a theorem-proving problem in pure first-order logic, thereby
 eliminating the need of generating further induction axioms.
However, coding introduces an overhead which, although constant, is so high that it dominates the complexity of
 all practically relevant instances to such an extent that this approach is not useful in practice.
(2)~Similar to but even simpler than~(1), one can just fix a $k$, sufficiently high for practical purposes, so that
 we are interested only in finding proofs whose induction axioms contain at most $k$ symbols.
Since there are only a finite number of such induction axioms, the theorem-proving problem, again, becomes a pure
 first-order problem.
Just as in~(1), although this avoids the need of generating further induction axioms, it introduces
 a syntactic overhead which renders this approach useless in practice (even though coding does
 not play a role here).

So, in practice, these brute force methods are not an option.
Instead, one typically tries to find {\em simple} induction axioms, tailored to the goal, and a proof
 of the goal based on them.
Again, we cannot claim to make this notion precise here but we refer to typical solutions of the
 TIP problems as listed, e.g., in~\cite{Ireland96Productive}.
We can now make the following computational observation:
\begin{compobs}
For sentences $\sigma$ which are input to an inductive theorem prover: if it is feasible to find
 simple formulas $\theta_1,\ldots,\theta_n$ such that
 $\PAminus + \{ \ind\theta_1,\ldots,\ind\theta_n \} \proves \sigma$, then it is feasible
 to find a proof in pure first-order logic of $\sigma$ from $\PAminus + \{ \ind\theta_1,\ldots,\ind\theta_n \}$.
\end{compobs}
It is important to note here that the input sentence $\sigma$ and the induction formulas $\theta_i$ are
 restricted as discussed above.
The word ``feasible'' is to be understood in the sense of the possibility of an implementation which solves
 the task successfully on contemporary hardware in a reasonable amount of time.

The observation is based on the following grounds. First,
automated theorem proving in pure first-order logic is a subject
that has undergone continuous progress for decades and has reached a quite mature state. The regular
CASC-competition~\cite{Sutcliffe16CASC} is a testament to that as is the widespread use of
first-order theorem provers in external tools, e.g., Sledgehammer~\cite{Paulson12Three}. Inductive theorem proving,
and in particular the generation of simple non-analytic induction invariants, does not enjoy a comparable
level of stability and maturity (yet?). Secondly, and in terms of concrete evidence, we have considered
53 proofs of problems from the TIP-library. These proofs have been manually entered in the
GAPT-system~\cite{Ebner16System} by a student of the first author (for another purpose). Of these
53 proofs, GAPT's built-in first-order prover Escargot, which is a quite simple superposition prover,
could re-prove 48 based on the induction axioms alone within a timeout of $1$ minute per proof using, on
average, $3.3$ seconds per proof on standard PC hardware.

Another way to put this computational observation is the following: as described in
 Section~\ref{sec.nec_non_an}, the search space in inductive theorem proving has two dimensions:
 (1)~the search for instances of the induction scheme which prove the goal,
 and (2)~the search for a proof in pure first-order logic of the goal from these instances.
The above computational observation then states that, for input sentences and induction formulas considered
 in practice, the search space extends so much more in the first dimension than in the second that we
 can afford to disregard the second.
This observation forms an important basis for our theoretical model of inductive theorem proving.

\subsection{One induction axiom is enough}

As a final step towards our model of inductive theorem proving, we will see in this section
that we can restrict our attention to the use of a single induction axiom.
\begin{defi}
An $\Lor$ formula $\phi(x)$ is called \defm{inductive} if $\PAminus\proves \phi(0)$
and $\PAminus\proves \fa x{(\phi(x)\then\phi(x+1))}$.
\end{defi}
If $\phi(x)$ is an inductive formula,
 then $\fa x{\phi(x)}$ is trivially equivalent over~$\PAminus$ to
  the induction axiom
  \begin{equation*}
   \phi(0)\wedge\fa x{\bigl(\phi(x)\then\phi(x+1)\bigr)}
   \then\fa x{\phi(x)}.
  \end{equation*}
Therefore,
 sentences of the form $\fa x{\phi(x)}$, where $\phi(x)$ is inductive,
  can be viewed as particular instances of parameter-free induction axioms.
Conversely, as Lemma~\ref{lem:ax=fma} below shows,
 every induction axiom is equivalent over~$\PAminus$
  to an induction axiom of this form.
As a result, every induction axiom corresponds to an inductive formula, and
 the full induction scheme is equivalent to its parameter-free counterpart.
The argument is presumably well known,
 cf.~Kaye~\cite[Exercise~8.3]{book:modelPA}.

\begin{lem}\label{lem:ax=fma}
Let $\theta(x,\bar z)$ be an $\Lor$~formula.
Define $\phi(x)$ to be
 \begin{equation*}
  \fa{\bar z}{\bigl(
   \theta(0,\bar z)
   \wedge\fa y{\bigl(\theta(y,\bar z)\then\theta(y+1,\bar z)\bigr)}
   \then\theta(x,\bar z)
  \bigr)}.
 \end{equation*}
Then $\phi(x)$ is inductive and
 \begin{math}
  \PAminus\proves\ind_x\theta\nsc\fa x{\phi(x)}
 \end{math}.
\end{lem}

\begin{proof}
Let us first verify that $\phi(x)$~is inductive.
Work over~$\PAminus$.
We have $\phi(0)$ trivially.
Suppose $x_0$~is such that $\phi(x_0)$~holds, and
 take $\bar z$~such that the hypothesis in~$\phi(x_0+1)$ holds,
  i.e.,
  \begin{equation*}
   \theta(0,\bar z)
   \wedge\fa y{\bigl(\theta(y,\bar z)\then\theta(y+1,\bar z)\bigr)}.
  \end{equation*}
Then $\theta(x_0,\bar z)$ must be true since~$\phi(x_0)$, and
 thus $\theta(x_0+1,\bar z)$ is also true by the second conjunct displayed above.
This shows $\phi(x_0+1)$.

Next, we verify that
 \begin{math}
  \PAminus\proves\ind_x\theta\nsc\fa x{\phi(x)}
 \end{math}.
Work over~$\PAminus$ again.
Suppose $\fa x{\phi(x)}$.
Take $\bar z$ such that
 \begin{displaymath}
  \theta(0,\bar z)
  \wedge\fa y{\bigl(\theta(y,\bar z)\then\theta(y+1,\bar z)\bigr)}.
 \end{displaymath}
We want $\fa x{\theta(x,\bar z)}$.
So pick any~$x_0$.
We know $\phi(x_0)$~holds by hypothesis.
Thus $\theta(x_0,\bar z)$ by the definition of~$\phi(x)$, as required.

Conversely, assume $\ind_x\theta$ holds, i.e.,
 \begin{equation*}
  \fa{\bar z}{\bigl(
   \theta(0,\bar z)
   \wedge\fa y{\bigl(\theta(y,\bar z)\then\theta(y+1,\bar z)\bigr)}
   \then\fa x{\theta(x,\bar z)}
  \bigr)}.
 \end{equation*}
Let $x_0$ be arbitrary.
We want~$\phi(x_0)$.
So take any $\bar z$ such that
 \begin{math}
  \theta(0,\bar z)
  \wedge\fa y{\bigl(\theta(y,\bar z)\then\theta(y+1,\bar z)\bigr)}
 \end{math}.
Then our assumption implies~$\theta(x_0,\bar z)$,
 which is what we want.
\end{proof}

The key idea behind the lemma above is that
 inductive formulas are, in a sense, closed under definable conjunction.
In particular, any two induction axioms are implied by a third over~$\PAminus$.

\begin{prop}[Gentzen~\cite{art:combine-indn}]\label{prop:PA-merge}
For all $\Lor$~formulas $\theta_0,\theta_1$,
 there is an inductive formula $\phi(x)$ such that
  \begin{math}
   \PAminus\proves\fa x{\phi(x)}\then\ind\theta_0\wedge\ind\theta_1
  \end{math}.
\end{prop}

\begin{proof}
Apply Lemma~\ref{lem:ax=fma} to find
 inductive formulas $\psi_0(x)$ and $\psi_1(x)$
  such that
  \begin{math}
   \PAminus\proves\ind\theta_i\nsc\fa x{\psi_i(x)}
  \end{math} for each $i<2$.
Define $\phi(x)=\psi_0(x)\wedge\psi_1(x)$.
As $\psi_0(x)$ and $\psi_1(x)$ are both inductive,
 it is easy to see that $\phi(x)$ is inductive too.
Moreover, the sentence $\fa x{\phi(x)}$
 implies $\fa x{\psi_0(x)}\wedge\fa x{\psi_1(x)}$
 and thus also $\ind\theta_0\wedge\ind\theta_1$ over~$\PAminus$.
\end{proof}
The above proof straightforwardly generalizes to an arbitrary number of $\Lor$~formulas $\theta_1,\ldots,\theta_n$ and so we obtain:
\begin{cor}\label{cor:ITP}
Let $\sigma$ be an $\Lor$~sentence.
Then $\PA\proves\sigma$ if and only if there is an inductive formula~$\phi(x)$
 such that $\PAminus\proves\fa x{\phi(x)}\then\sigma$. \qed
\end{cor}
Since we only use simple syntactic operations (conjunction in Corollary~\ref{cor:ITP} and
the definition in the statement of Lemma~\ref{lem:ax=fma}) to combine many induction axioms
into one, the computational observation is preserved even when restricted to a single induction axiom.
Our theoretical model for inductive theorem proving is now the following computational problem.
\cprob{\ITP}
{A sentence $\sigma$ provable in $\PA$}
{An inductive formula $\varphi(x)$ s.t.\ $\PAminus \proves \fa x{\phi(x)} \then \sigma$}
Formally this just defines a binary relation (between $\sigma$ and $\varphi(x)$).
But of course, implicitly, we take the perspective of wanting to compute such a $\varphi(x)$ from
 a given $\sigma$.
We claim that this computational problem is a suitable theoretical model of the practice of
 inductive theorem proving.
This claim rests on the computational observation made in Section~\ref{sec:compass} and sharpened in this section.
In the rest of this paper we will study this problem, in particular
 by relating it to several of its variants.

\section{Variations of the proof shape}\label{sec:proof_shapes}

The {\ITP} problem as defined above induces a natural proof shape: the combination
of (i) a $\PAminus$-proof of the induction base, (ii) a $\PAminus$-proof of the induction
step, and (iii) a $\PAminus$-proof of $\fa x{\phi(x)}\then\sigma$. As long as we use ordinary
successor induction, there is no freedom in the first two proof obligations, there is however in
the third. In this section we will
consider two variants of {\ITP} which are obtained by modifying (iii).

The sharp-eyed reader may have noticed that
 the inductive formula $\phi(x)$ in our proof of Proposition~\ref{prop:PA-merge}
  actually makes
   \begin{math}
    \PAminus\proves\fa x{\phi(x)}\nsc\ind\theta_0\wedge\ind\theta_1
   \end{math}.
We can use this to obtain
\begin{thm}\label{thm:Th(PA)=indn}
Let $\sigma$ be an $\Lor$~sentence.
Then $\PA\proves\sigma$ if and only if there is an inductive formula~$\phi(x)$
 such that $\PAminus\proves\fa x{\phi(x)}\nsc\sigma$.
\end{thm}

\begin{proof}
The ``if'' direction is clear from the definition of inductive formulas.
For the ``only if'' direction,
 apply Proposition~\ref{prop:PA-merge} to find an inductive formula~$\psi(x)$
  such that $\PAminus\proves\fa x{\psi(x)}\then\sigma$.
We verify that
 \begin{equation*}
  \phi(x)\qequals\neg\sigma\then\psi(x)
 \end{equation*} has the properties we want.
First, it is clear that $\PAminus\proves\sigma\then\fa x{\phi(x)}$.
Second, work over~$\PAminus$, and suppose $\neg\sigma$.
Since $\PAminus\proves\fa x{\psi(x)}\then\sigma$,
 this implies $\ex x{\neg\psi(x)}$.
If $x_0$ is such that $\neg\psi(x_0)$, then $\neg\sigma\wedge\neg\psi(x_0)$
 and so $\neg\phi(x_0)$.
We can thus conclude $\PAminus\proves\neg\sigma\then\neg\fa x{\phi(x)}$.
Finally, the formula $\phi(x)$ is inductive
 because it is equivalent to either $x=x$ or $\psi(x)$
   depending on whether $\sigma$ holds or not, and
  both $x=x$ and~$\psi(x)$ are inductive.
\end{proof}

This is a particular case of a more general phenomenon.
As a normal form theorem in a very broad sense,
 it is perhaps reminiscent of
  the Friedman--Goldfarb--Harrington Theorem
    and its generalizations~\cite{art:FaithFalsity,inproc:Tjump=prov},
   which assert that over a sufficiently strong base theory,
    every $\Lor$~sentence is equivalent to a consistency statement.
This result motivates the consideration of the equivalence version of \ITP:
\cprob{\ITPEq}
{A sentence $\sigma$ provable in $\PA$}
{An inductive formula $\varphi(x)$ s.t.\ $\PAminus \proves \fa x{\phi(x)} \nsc \sigma$}
Theorem~\ref{thm:Th(PA)=indn} ensures that, just as {\ITP}, also {\ITPEq} is a total relation in the sense that for every input
there is an output as required. However, for a fixed $\sigma$ the $\phi(x)$'s permitted in {\ITPEq}
are a strongly restricted subset of those permitted in {\ITP}. This leads to a significant
reduction of the search space. Typically, restrictions of the search space play a crucial role
for automated theorem proving in practice. We are not aware of a technique that would exploit
this reduction of {\ITP} to {\ITPEq}. In how far this restriction to equivalent formulas
is useful in practice therefore remains unclear for the time being.

Another modification of the proof shape of {\ITP} consists of considering a universally quantified
$\sigma$, i.e., $\sigma = \fa x{\psi(x)}$, and seeking an inductive formula~$\phi(x)$ such that
\begin{math}
\PAminus\proves\fa x{\bigl(\phi(x)\then\psi(x)\bigr)}
\end{math}.
This is of relevance to computer science since it corresponds
 to the treatment of loops in correctness proofs for imperative programs
 by loop invariants as in the Hoare calculus~\cite{Apt10Verification,Bradley07Calculus}.
As one may expect, this method does not work for all $\PA$-provable formulas.

\begin{prop}\label{prop:unif-ncompl}
There is an $\Lor$~formula~$\psi(x)$ such that
 $\PA\proves\fa x{\psi(x)}$
 but no inductive formula $\phi(x)$
  makes $\PAminus\proves\fa x{\bigl(\phi(x)\then\psi(x)\bigr)}$.
\end{prop}

\begin{proof}
Pick an $\Lor$~sentence $\sigma$ that
 is provable in~$\PA$ but not in~$\PAminus$.
Consider the $\Lor$~formula $\psi(x)$, which is defined to be
 \begin{equation*}
  \sigma\lor x\not=0.
 \end{equation*}
Then $\PA\proves\fa x{\psi(x)}$
 because $\PA\proves\sigma$.
Let $M\models\PAminus+\neg\sigma$,
 which exists since $\PAminus\nproves\sigma$.
Then $M\models\neg\psi(0)$ by the definition of~$\psi(x)$.
Therefore, for \emph{no} formula~$\phi(x)$ can
 \begin{equation*}
  \PAminus\proves\fa x{(\phi(x)\then\psi(x))}\wedge\phi(0). 
  \tag*{\qEd}
 \end{equation*}
 \def\popQED{}
\end{proof}

As a result, the following uniform version of \ITP
\cprob{\ITPU}
{A sentence $\fa x{\psi(x)}$ provable in $\PA$}
{An inductive formula $\phi(x)$ s.t.\ $\PAminus \proves \fa x{(\phi(x) \then \psi(x))}$}
is not a total relation.
Nevertheless, some weak form of completeness is still possible
 if we restrict ourselves to simple enough $\Lor$~sentences provable
  in a sufficiently weak fragment of~$\PA$,
 as the following theorem shows.

\begin{defi}
An $\Lor$~formula is \defm{bounded} if
 all the quantifiers it contains are of the form
  $\falt xt{}$ or $\exlt xt{}$,
   where $t$~is a term in~$\Lor$ that does not involve the variable~$x$.
Bounded formulas are also called \defm{$\Delta_0$~formulas}.
The theory \defm{$\ind\Delta_0$} is
\[
  \PAminus
  +\{\ind_x\theta:\text{$\theta(x,\bar z)$ is a bounded $\Lor$~formula}\}.
  \]
Fix a bounded formula $y=2^x$ such that
 \begin{equation*}
  \ind\Delta_0\proves\begin{aligned}[t]
   &\fa{x,y,y'}{(y=2^x\wedge y'=2^x\then y=y')}\\
   &\wedge2^0=1\wedge\fa{x,y}{(y=2^x\nsc2y=2^{x+1})}.
  \end{aligned}
 \end{equation*}
Let \defm{$\exp$} be the axiom $\fa x{\ex y{(y=2^x)}}$.
\end{defi}
See Section~V.3(c) in H\'ajek--Pudl\'ak~\cite{book:hajek+pudlak},
  for example,
 for a construction of the formula $y=2^x$.

\begin{thm}[Wilkie--Paris]\label{thm:D0-n-uni}
The following are equivalent for a bounded formula~$\psi(x)$.
\begin{enumerate}\tfaeenum
\item $\ind\Delta_0+\exp\proves\fa x{\psi(x)}$.
\item There is an inductive formula~$\phi(x)$
  such that $\PAminus\proves\fa x{(\phi(x)\then\psi(x))}$.
\end{enumerate}
\end{thm}

\begin{proof}
See Corollary~8.7 in Wilkie--Paris~\cite{art:indn_bddarith}
 or Theorem~V.5.26 in H\'ajek--Pudl\'ak~\cite{book:hajek+pudlak}.
\end{proof}

\section{Different forms of induction}\label{sec:forms_of_induction}

In our definition of {\ITP} we have fixed the induction scheme to the ordinary successor
induction. This is by no means the only choice. It is well known that the induction
scheme has many equivalent formulations. In this section we will introduce those alternative
formulations that we treat in this paper, and start to study their relationship, both in the
general and in the quantifier-free setting. Our interest in the quantifier-free setting is motivated
by the fact that some methods for inductive theorem proving restrict themselves to
quantifier-free induction formulas, see, e.g.~\cite{Boyer79Computational}.

\subsection{Different induction schemes}

\begin{defi}
Let $\theta(x,\bar z)$ be an $\Lor$~formula. We define the following induction axioms:
\begin{itemize}
\item The \emph{$<$-induction axiom} $\ind^<_x \theta$ is
\begin{equation*}
\fa{\bar z}{\Bigl(
  \fa y{\bigl(\falt xy{\theta(x, \bar z)\then\theta(y, \bar z)}\bigr)} \then \fa x{\theta(x,\bar z)}
\Bigr)}.
\end{equation*}

\item Let $n\in\IN$.  The \emph{$(n+1)$-step induction axiom} $\ind^\text{$(n+1)$-step}_x \theta$ is
\begin{equation*}
\fa{\bar z}{\Bigl(
  \bigwwedge_{k<n+1}\theta(k,\bar z)
  \wedge\fa x{\bigl(\theta(x, \bar z)\then\theta(x+n+1,\bar z)\bigr)} \then \fa x{\theta(x,\bar z)}
\Bigr)}.
\end{equation*}

\item Let $n\in\IN$.  The \emph{$(n+1)$-induction axiom} $\ind^{n+1}_x \theta$ is
\begin{equation*}
\fa{\bar z}{\Bigl(
  \bigwwedge_{k<n+1}\theta(k, \bar z)
  \wedge\fa x{\Bigl(\bigwwedge_{k<n+1}\theta(x+k, \bar z)\then\theta(x+n+1, \bar z)\Bigr)} \then \fa x{\theta(x,\bar z)}
\Bigr)}.
\end{equation*}

\item The \emph{polynomial induction axiom} $\ind^{\text{\p}}_x \theta$ is
\begin{equation*}
\fa{\bar z}{\Bigl(
  \theta(0,\bar z) \wedge\fa x{\bigl(\theta(x,\bar z)\then\theta(2x, \bar z)\wedge\theta(2x+1,\bar z)\bigr)} \then \fa x{\theta(x,\bar z)}
\Bigr)}.
\end{equation*}
\end{itemize}
\end{defi}
Each of these schemes naturally induces a notion of inductiveness.
\begin{defi}
Let $\phi(x)$ be an $\Lor$~formula with precisely one free variable~$x$.
\begin{itemize}
\item $\phi(x)$ is \defm{$<$-inductive} if
  \begin{math}
   \PAminus\proves\fa y{\bigl(\falt xy{\phi(x)\then\phi(y)}\bigr)}
  \end{math}.
\item Let $n\in\IN$.
 We say $\phi(x)$ is \defm{$(n+1)$-step inductive} if
 \begin{equation*}
  \PAminus\proves\bigwwedge_{k<n+1}\phi(k)
   \wedge\fa x{\bigl(\phi(x)\then\phi(x+n+1)\bigr)}.
 \end{equation*}
\item Let $n\in\IN$.
 We say $\phi(x)$ is \defm{$(n+1)$-inductive} if
 \begin{equation*}
  \PAminus\proves\bigwwedge_{k<n+1}\phi(k)
   \wedge\fa x{\Bigl(\bigwwedge_{k<n+1}\phi(x+k)\then\phi(x+n+1)\Bigr)}.
 \end{equation*}
\item $\phi(x)$ is \defm{polynomially inductive},
  or simply \defm{\p-inductive},
 if
 \begin{equation*}
  \PAminus\proves\phi(0)
   \wedge\fa x{\bigl(\phi(x)\then\phi(2x)\wedge\phi(2x+1)\bigr)}.
 \end{equation*}
\end{itemize}
\end{defi}

The first of the above formulations, $<$-induction, is also known as strong induction and is commonly used
in mathematics. The consideration of $(n+1)$-step induction is motivated by its use
in~\cite{Bundy05Rippling}. The $k$-induction scheme has become popular in computer-aided
verification, see, e.g.,~\cite{Sheeran00Checking,Donaldson11Automatic}. Polynomial induction
has been introduced by Buss in~\cite{book:bussPhD} for the study of weak arithmetical theories
and their relation to computational complexity classes.
In this paper, we restrict ourselves to the base-2 polynomial induction scheme
 in order to keep the exposition sufficiently simple;
 see Remark~\ref{rmk:bases} for some extra information about other bases.
\begin{prop}\label{prop:altPA}
The following are equivalent for every $n\in\IN$:
\begin{enumerate}[leftmargin=8mm] \tfaeenum
\item $\PA = \PAminus + \{\ind_x\theta:
 \text{$\theta(x,\bar z)$ is an $\Lor$~formula}
\}$;\label{item:altPA/PA}

\item $\PAminus + \{\ind^<_x\theta:
 \text{$\theta(x,\bar z)$ is an $\Lor$~formula}
\}$;\label{item:altPA/'}

\item $\PAminus + \{\ind^\text{\rm$(n+1)$-step}_x\theta:
 \text{$\theta(x,\bar z)$ is an $\Lor$~formula}
\}$;\label{item:altPA/step}

\item $\PAminus + \{\ind^{n+1}_x\theta:
 \text{$\theta(x,\bar z)$ is an $\Lor$~formula}
\}$;\label{item:altPA/n}

\item $\PAminus + \{\ind^\text{\p}_x\theta:
 \text{$\theta(x,\bar z)$ is an $\Lor$~formula}
\}$.\label{item:altPA/p}
\end{enumerate}
\end{prop}

The proof of Proposition~\ref{prop:altPA} is a straightforward exercise.
So instead of showing it here, we show its quantifier-free counterpart,
Theorem~\ref{thm:IOpen}.
Everything carries over \emph{except} the arguments
 for $\itemref{item:altPA/PA}\Then\itemref{item:altPA/'}$ and
     $\itemref{item:altPA/p}\Then\itemref{item:altPA/PA}$.
The former of these implications was shown to remain true
 by a more substantial argument in Shepherdson~\cite{art:sheph/model}.
First we gather a few statements easily provable in $\PAminus$ that will
be useful at several occasions.

\begin{lem}\label{lem:PA-}
$\PAminus$ proves
\begin{enumerate}\partenum
\item $\fa x{\fa y{\fa z{(x+z<y+z\then x<y)}}}$,\label{part:PA-/<+cancel}
\item $\fa x{\fa y{\fa z{(x\times z<y\times z\then x<y)}}}$,\label{part:PA-/<xcancel}
\item $\fa x{\fa y{(x<y\then x+1\leq y)}}$,\label{part:PA-/discrete}
\item $\fa x{\bigl(x\leq n\then\bigvvee_{k\leq n}x=k\bigr)}$
  for every $n\in\IN$,\label{part:PA-/ee}
\item $\fa x{(x<x+1)}$, and\label{part:PA-/x<x+1}
\item $\fa x{(x\not=0\then\ex y{(x=y+1)})}$.\label{part:PA-/pred}
\end{enumerate}
\end{lem}

\begin{proof}
For \partref{part:PA-/<+cancel} and~\partref{part:PA-/<xcancel},
 see the top of page~18 in Kaye~\cite{book:modelPA}.
For \partref{part:PA-/discrete} and~\partref{part:PA-/ee},
 see Proposition~2.1 and Lemma~2.6 in Kaye~\cite{book:modelPA} respectively.

For \partref{part:PA-/x<x+1},
 we see that $x+1=(x+0)+1$ by axiom~\ref{ax:PA-/+0}.
So we are done by axiom~\ref{ax:PA-/<}.

For \partref{part:PA-/pred}, work over~$\PAminus$.
If $x\not=0$, then axiom~\axref{ax:PA-/nneg} implies $x>0$,
 and so we get the~$y$ we want by~\axref{ax:PA-/<}.
\end{proof}

\begin{thm}[mostly Shepherdson]\label{thm:IOpen}
The following are equivalent for all $n\in\IN$:
\begin{enumerate}[leftmargin=8mm] \tfaeenum
\item $\IOpen = \PAminus + \{\ind_x\theta:
 \text{$\theta(x,\bar z)$ is a quantifier-free $\Lor$~formula}
\}$;\label{item:IOpen/IOpen}

\item $\PAminus + \{\ind^<_x\theta:
 \text{$\theta(x,\bar z)$ is a quantifier-free $\Lor$~formula}
\}$;\label{item:IOpen/'}

\item $\PAminus + \{\ind^\text{\rm$(n+1)$-step}_x\theta:
 \text{$\theta(x,\bar z)$ is a quantifier-free $\Lor$~formula}
\}$;\label{item:IOpen/step}

\item $\PAminus + \{\ind^{n+1}_x\theta:
 \text{$\theta(x,\bar z)$ is a quantifier-free $\Lor$~formula}
\}$.\label{item:IOpen/n}
\end{enumerate}
\end{thm}

\begin{proof}
We refer the reader to the original Shepherdson paper~\cite{art:sheph/model}
 for a proof of $\itemref{item:IOpen/IOpen}\Then\itemref{item:IOpen/'}$.
We do not include a proof here
 because it would be too distracting for the present paper
  to set up all the algebraic materials for the argument.
Nevertheless, some of the ideas can be found in Section~\ref{sec:order}.

Consider $\itemref{item:IOpen/'}\Then\itemref{item:IOpen/step}$.
Work over~\itemref{item:IOpen/'}.
Fix $\bar z$ and let $\theta(x,\bar z)$ be a quantifier-free $\Lor$~formula
 such that $\ex x{\neg\theta(x,\bar z)}$.
Use~\itemref{item:IOpen/'} to find $x_0=(\min x)(\neg\theta(x,\bar z))$.
If $x_0<n+1$, then $\bigvvee_{k<n+1}\neg\theta(k,\bar z)$
 by Lemma~\ref{lem:PA-}\partref{part:PA-/ee}.
So suppose $x_0\geq n+1$.
Then $x_0>n$
 by (the contrapositive of) Lemma~\ref{lem:PA-}\partref{part:PA-/discrete}.
Apply axiom~\axref{ax:PA-/<} to find $w_0$
 such that $x_0=n+w_0+1>w_0$.
Then $\theta(w_0,\bar z)$ holds by the minimality of~$x_0$,
 but $\neg\theta(w_0+n+1,\bar z)$, as required.

The implication $\itemref{item:IOpen/step}\Then\itemref{item:IOpen/n}$
 is clear.

Consider $\itemref{item:IOpen/n}\Then\itemref{item:IOpen/IOpen}$.
Work over~\itemref{item:IOpen/n}.
Fix $\bar z$ and let $\theta(x,\bar z)$ be a quantifier-free $\Lor$~formula
 such that
 \begin{math}
  \theta(0,\bar z)
  \wedge\fa x{\bigl(\theta(x,\bar z)\then\theta(x+1,\bar z)\bigr)}
 \end{math}.
Then
 \begin{equation*}
  \bigwwedge_{k<n+1}\theta(k,\bar z)
  \wedge\fa x{\Bigl(
   \bigwwedge_{k<n+1}\theta(x,\bar z)\then\theta(x+n+1,\bar z)
  \Bigr)},
 \end{equation*} and so we are done by~\itemref{item:IOpen/n}.
\end{proof}

The following small lemma will be handy when proving the quantifier-free analogue of
 $\itemref{item:altPA/PA}\Then\itemref{item:altPA/p}$ of Proposition~\ref{prop:altPA}.

\begin{lem}\label{lem:Div}
\begin{enumerate}\partenum
\item $\IOpen\proves\fa{x,d}{\ex{q,r}{(d\not=0\then x=qd+r\wedge r<d)}}$, but
\item $\PAminus\nproves\fa x{\ex y{(x=2y\vee x=2y+1)}}$.
\end{enumerate}
\end{lem}

\begin{proof}
For the provability part, we follow Kaye~\cite[page~5]{art:IDio}.
Work over $\IOpen$.
Take any $x,d$ with $d\not=0$.
Using the axioms of $\PAminus$
  and Lemma~\ref{lem:PA-}\partref{part:PA-/pred},
 one can routinely verify that $0d\leq x<(x+1)d$.
Apply induction on~$q$ for the atomic formula $qd\leq x$ to
 find $q$ which satisfies $qd\leq x<(q+1)d=qd+d$.
Then setting $r=x-qd<qd+d-qd=d$ gives us what we want.

For the unprovability part,
 consider the set $\ZZ[X]^+$ of all elements of the polynomial ring~$\ZZ[X]$
  that either are zero or have positive leading coefficients.
It is naturally a model of~$\PAminus$,
 as the reader can verify~\cite[Section~2.1]{book:modelPA}.
Clearly,
 \begin{equation*}
  \ZZ[X]^+\models\neg\ex y{(X=2y\vee X=2y+1)}. 
  \tag*{\qEd}
 \end{equation*}
 \def\popQED{}
\end{proof}

\begin{prop}
$\IOpen$ proves
 \begin{equation*}
  \fa{\bar z}{\bigl(
   \theta(0,\bar z)\wedge\fa x{\bigl(
    \theta(x,\bar z)\then\theta(2x,\bar z)\wedge\theta(2x+1,\bar z)
   \bigr)}
   \then\fa x{\theta(x,\bar z)}
  \bigr)}
 \end{equation*}
 for every quantifier-free $\Lor$~formula $\theta(x,\bar z)$.
\end{prop}

\begin{proof}
Work over $\IOpen$.
Let $\theta(x,\bar z)$ be a quantifier-free $\Lor$~formula and
 $\bar z$ be parameters
 such that
 \begin{math}
  \ex x{\neg\theta(x,\bar z)}
 \end{math}.
By~Proposition~\ref{thm:IOpen}\partref{item:IOpen/'},
 this has a least witness, say~$x_0$.
If $x_0=0$, then $\neg\theta(0,\bar z)$.
So suppose $x_0\not=0$.
Using Lemma~\ref{lem:Div},
 find $y_0$ such that $x_0=2y_0$ or $x_0=2y_0+1$.
Since $x_0\not=0$, we know $y_0<x_0$
 by Lemma~\ref{lem:PA-}\partref{part:PA-/pred} and axiom~\axref{ax:PA-/<}.
So $\theta(y_0,\bar z)$ by the minimality of~$x_0$.
We are thus done
 because $\neg\theta(2y_0,\bar z)\vee\neg\theta(2y_0+1,\bar z)$.
\end{proof}

Following the proof of Theorem~V.4.6
  in H\'ajek--Pudl\'ak~\cite{book:hajek+pudlak},
 one can prove $\IOpen$ from $\PAminus$ plus the polynomial induction scheme
  for formulas of the form $\falt yt{\eta(x,y,\bar z)}$
   where $\eta(x,y,\bar z)$ is a quantifier-free $\Lor$~formula and
         $t$~is an $\Lor$~term not involving~$y$.
The use of this extra bounded quantifier cannot be eliminated.

\begin{thm}[Johannsen]\label{thm.johannsen}
\begin{math}
\PAminus+\{\ind^{\text{\p}}_x\theta:
 \text{$\theta(x,\bar z)$ is a quantifier-free $\Lor$~formula}
\}
\end{math} does not prove
\begin{math}
\fa x{\ex y{(x=3y\vee x=3y+1\vee x=3y+2)}}
\end{math}.
Consequently, it is strictly weaker than~$\IOpen$.
\end{thm}

\begin{proof}
This follows from Theorem~2 in Johannsen~\cite{inproc:weakness-sbdd-pind}.
\end{proof}

\subsection{Variations of induction scheme and proof shape}

All results of Section~\ref{sec:proof_shapes},
   with the possible exception of Theorem~\ref{thm:D0-n-uni},
 actually remain true when the induction scheme is replaced
  by another scheme in Proposition~\ref{prop:altPA}.
The proofs are straightforward modifications of what we presented there, and
 hence are left to the reader.
In this section we abstract one part that may be of independent interest.
This stems from the observation that, in a sense,
 Proposition~\ref{prop:PA-merge} is the only property of~$\PA$
  one needs in establishing Theorem~\ref{thm:Th(PA)=indn}.

\begin{defi}
Let $\lang$ be a language.
Denote by \defm{$\lang(\newX)$} the language obtained from~$\lang$
 by adding one new unary predicate symbol~$\newX$.
Let $S$ be a sentence in~$\lang(\newX)$.
If $\phi(x,\bar z)$ is an $\lang$~formula,
 then define \defm{$S\phi$} to be
  the universal closure of the $\lang$~formula
   obtained from~$S$ by replacing
    each occurrence of $\newX(\dots)$ by~$\phi(\dots,\bar z)$.
The \defm{scheme determined by~$S$}, denoted \defm{$S\lang$},
 is defined to be
 \begin{math}
  \{S\phi:\phi\in\lang\}
 \end{math}.
The scheme $S\lang$ is \defm{mergeable} over an $\lang$~theory~$B$
 if for all $\lang$~formulas $\psi_0,\psi_1$,
  there is an $\lang$~formula $\phi$ such that
  \begin{math}
   B+S\phi\proves S\psi_0\wedge S\psi_1
  \end{math}.
\end{defi}

Proposition~\ref{prop:PA-merge} demonstrates the mergeability of
 the successor induction scheme over~$\PAminus$.
A very similar proof shows that
 all other schemes in Proposition~\ref{prop:altPA} are also mergeable over~$\PAminus$.
Another example is the comprehension scheme in second-order arithmetic,
 which is mergeable over rather weak base theories.

\begin{thm}\label{thm:merge}
Fix a language~$\lang$.
Let $S\lang$ be a scheme and $B$~be an $\lang$~theory
 such that $B\proves S\top$.
The following are equivalent.
\begin{enumerate}\tfaeenum
\item $S\lang$ is mergeable over~$B$.\label{cond:merge/merge}
\item For every $\lang$~sentence~$\sigma$ provable from $B+S\lang$,
  there exists an $\lang$~formula~$\phi(x,\bar z)$
   such that $B\proves S\phi\nsc\sigma$.\label{cond:merge/rev}
\end{enumerate}
\end{thm}

\begin{proof}
For $\condref{cond:merge/merge}\Then\condref{cond:merge/rev}$,
 imitate the proof of Theorem~\ref{thm:Th(PA)=indn}.
Conversely, suppose \condref{cond:merge/rev}~holds.
Pick arbitrary $\lang$~formulas $\psi_0$ and $\psi_1$.
Define $\sigma=S\psi_0\wedge S\psi_1$.
Then $B+S\lang\proves\sigma$ trivially.
By \condref{cond:merge/rev}, we get an $\lang$~formula $\phi$
 such that $B\proves S\phi\nsc\sigma$ and
  hence $B+S\phi\proves S\psi_0\wedge S\psi_1$, as required.
\end{proof}

\begin{cor}
Let $n\in\IN$ and let $\sigma$ be an $\Lor$~sentence.
The following are equivalent.
\begin{enumerate}[leftmargin=8mm] \tfaeenum
\item $\PA\proves\sigma$.
\item There is a $<$-inductive formula $\phi(x)$
  such that $\PAminus\proves\fa x{\phi(x)}\nsc\sigma$.
\item There is an $(n+1)$-step inductive formula $\phi(x)$
  such that $\PAminus\proves\fa x{\phi(x)}\nsc\sigma$.
\item There is an $(n+1)$-inductive formula $\phi(x)$
  such that $\PAminus\proves\fa x{\phi(x)}\nsc\sigma$.
\item There is a \p-inductive formula $\phi(x)$
  such that $\PAminus\proves\fa x{\phi(x)}\nsc\sigma$. \qed
\end{enumerate}
\end{cor}

\subsection{Walther's method for comparing induction schemes}\label{sec:walther}

In the context of inductive theorem proving, Walther~\cite[Section~7]{inproc:computindax} proposed a method
 of comparing induction axioms.
Let us formulate his method as follows.
If $B$~is a finite set of natural numbers and
   $S$~is a finite set of $\Lor$~terms,
 then define \defm{$\PA(B,S)$} to be the $\Lor$~theory axiomatized by~$\PAminus$
  and the following axioms for all $\Lor$~formulas~$\theta(x,\bar z)$:
   \begin{equation*}
    \fa{\bar z}{\Bigl(
     \bigwwedge_{k\in B}\theta(k,\bar z)
     \wedge\bigwwedge_{t\in S}\fa x{\bigl(
      \theta(x,\bar z)\then\theta(t(x,\bar z),\bar z)
     \bigr)}
     \then\fa x{\theta(x,\bar z)}
    \Bigr)}.
   \end{equation*}
In this terminology,
 the $(n+1)$-step induction scheme
  in Proposition~\ref{prop:altPA}\itemref{item:altPA/step} is essentially
  \begin{equation*}
   \PA(\{k\in\IN:k<n+1\},\{x+n+1\}).
  \end{equation*}
(Some choices of $B$ and~$S$
 may give rise to a scheme $\PA(B,S)$ that is not true in~$\IN$,
  but this does not concern us here.)
Walther observed that if $B\subseteq B'$ and $S\subseteq S'$,
 then $\PA(B,S)\proves\PA(B',S')$.

In this section we make the observation (not stated in~\cite{inproc:computindax}) that
this method for comparing induction schemes is incomplete
 in the sense that the converse implication is not true.
To see this, take $m,n\in\IN$ such that $m<n$.
Then
 \begin{align*}
  &\PA(\{k\in\IN:k<n+1\},\{x+m+1\})\\
  &\equiv\PA(\{k\in\IN:k<m+1\},\{x+m+1\})\\
  &\equiv\PA(\{k\in\IN:k<n+1\},\{x+n+1\})
 \end{align*} by Proposition~\ref{prop:altPA}.
Clearly $\{x+m+1\}\not\subseteq\{x+n+1\}$.

Thanks to Theorem~\ref{thm:IOpen},
 the observation in the previous paragraph has an analogue
  in the quantifier-free context.
We formulate this as follows.
If $B$~is a finite set of natural numbers and
   $S$~is a finite set of $\Lor$~terms,
 then define \defm{$\IOpen(B,S)$} to be the $\Lor$~theory axiomatized by~$\PAminus$
  and the following axioms for all quantifier-free $\Lor$~formulas~$\theta(x,\bar z)$:
   \begin{equation*}
    \fa{\bar z}{\Bigl(
     \bigwwedge_{k\in B}\theta(k,\bar z)
     \wedge\bigwwedge_{t\in S}\fa x{\bigl(
      \theta(x,\bar z)\then\theta(t(x,\bar z),\bar z)
     \bigr)}
     \then\fa x{\theta(x,\bar z)}
    \Bigr)}.
   \end{equation*}
Clearly,
 if $B\subseteq B'$ and $S\subseteq S'$,
  then $\IOpen(B,S)\proves\IOpen(B',S')$.
The converse is not true
 because whenever $m,n\in\IN$ such that $m<n$,
  we have
  \begin{align*}
   &\IOpen(\{k\in\IN:k<n+1\},\{x+m+1\})\\
   &\equiv\IOpen(\{k\in\IN:k<m+1\},\{x+m+1\})\\
   &\equiv\IOpen(\{k\in\IN:k<n+1\},\{x+n+1\})
  \end{align*} by Theorem~\ref{thm:IOpen},
  but $\{x+m+1\}\not\subseteq\{x+n+1\}$.

\section{Non-closure properties of cuts}\label{sec:non-closure_of_cuts}

In this short technical section we will establish an auxiliary result on
 cuts that will be used in Section~\ref{sec:comparing_solutions} and
 Section~\ref{sec:comparing_notions} for comparing solutions to {\ITP }
 and notions of inductiveness respectively.
Note that the meaning of the word ``cut'' here
 is different from that in Section~\ref{sec.nec_non_an}.
In the study of weak theories of arithmetic,
 these cuts are frequently used in interpretations.
They also serve as notions of smallness in many arguments.
In the following,
 we borrow some terminology from Visser~\cite{art:corto-basso}.
\begin{defi}
A \defm{cut} is an inductive formula~$\phi(x)$ such that
 \begin{equation*}
  \PAminus\proves\fa{x,y}{\bigl(x<y\wedge\phi(y)\then\phi(x)\bigr)}.
 \end{equation*}
An \defm{a-cut} is a cut~$\phi(x)$ such that
 \begin{equation*}
  \PAminus\proves\fa x{\bigl(\phi(x)\then\phi(x+x)\bigr)}.
 \end{equation*}
An \defm{am-cut} is an a-cut~$\phi(x)$ such that
 \begin{equation*}
  \PAminus\proves\fa x{\bigl(\phi(x)\then\phi(x\times x)\bigr)}.
 \end{equation*}
\end{defi}

Inductive formulas, cuts, a-cuts, and am-cuts are all different notions.
This fact seems to be well known,
 but we can find no proof of this in the literature.
So we include one here.
The proof assumes familiarity
 with some parts of the H\'ajek--Pudl\'ak book~\cite{book:hajek+pudlak}.

\begin{lem}\label{lem:+nx}
There are $\Lor$~formulas $\phi(x)$ and~$\delta(x)$ such that
\begin{enumerate}\propenum
\item $\phi(x)$ is an a-cut but not an am-cut;\label{cond:+nx/cut}
\item $\PAminus\proves\exi x{\delta(x)}
       \wedge\fa x{\bigl(\delta(x)\then\phi(x)\bigr)}$; and\label{cond:+nx/c-in-phi}
\item $\PAminus\proves\neg\fa x{\bigl(\phi(x)\then\phi(x\times x)\bigr)}
       \then\fa x{\bigl(\delta(x)\then\phi(x)\wedge\neg\phi(x\times x)\bigr)}$.\label{cond:+nx/c2>phi}
\end{enumerate}
\end{lem}

\begin{proof}
Let $M\models\PA$ which
 contains a nonstandard element~$a$ that is definable by a bounded formula.
For instance, if $M\models\neg\Con(\PA)$,
 then we can take $a$ to be
  the least (code of a) proof of contradiction from~$\PA$ in~$M$.
Fix a bounded formula $\delta_0(x)$ such that
 \begin{math}
  M\models\exi x{\delta_0(x)}\wedge\delta_0(a)
 \end{math}.
Notice
 \begin{equation*}
  I=2^{\IN a}=\{x\in M:\text{$x<2^{na}$ for some $n\in\IN$}\}
 \end{equation*} is an initial segment of~$M$,
 and $2^{ma}\times 2^{na}=2^{(m+n)a}\in I$ whenever $m,n\in\IN$.
So $I$ is itself an $\Lor$~structure.
Since $\delta_0(x)$ is a bounded formula,
 it is straightforward~\cite[Remark~IV.1.18]{book:hajek+pudlak} to check that
  $a$~is the unique element which satisfies $\delta_0(x)$ in~$I$.
We know $I\models\ind\Delta_0$
 because $\ind\Delta_0$ is preserved under
  taking initial segments~\cite[Remark~IV.1.21]{book:hajek+pudlak}.
Thus \begin{equation*}
 \log I=\{x\in I:I\models\ex u(2^x=u)\}
\end{equation*} is an initial segment of~$I$ that is closed under~$+$
 because
 \begin{align*}
  \ind\Delta_0
  &\proves
   \fa{x,y}{\bigl(
    \ex u{(2^x=u)}\wedge y\leq x\then\exle vu{(2^y=v)}
   \bigr)}\\
  &\phantom{{}\proves{}}
   \wedge\fa{u,v,x,y}{(
    2^x=u\wedge2^y=v\then2^{x+y}=uv
   )},
 \end{align*}
 as verified in H\'ajek--Pudl\'ak~\cite[Lemma~V.3.8(iii)
                                        and page~302]{book:hajek+pudlak}.
We know $\log I$ contains~$a$ but not~$a^2$ because
 \begin{equation*}
  \log I=\{x\in I:\text{$x<na$ for some $n\in\IN$}\}
 \end{equation*} and $a\not\in\IN$.
In particular, it is not closed under~$\times$.
The whole situation in~$I$ can be captured by the sentence
 \begin{equation*}
  \sigma\qequals\begin{aligned}[t]
   &\phi_0(0)\wedge\fa x{(\phi_0(x)\then\phi_0(x+x))}
    \wedge\fa{x,y}{(x<y\wedge\phi_0(y)\then\phi_0(x))}\\
   &\wedge\exi x{\delta_0(x)}
    \wedge\ex x{(\delta_0(x)\wedge\phi_0(x)\wedge\neg\phi_0(x\times x))},
  \end{aligned}
 \end{equation*}
 where $\phi_0(x)$ denotes the formula $\ex u(2^x=u)$.
Clearly,
 \begin{align*}
    \phi(x)&\qequals\sigma\then\phi_0(x)\qand\\
  \delta(x)&\qequals(\sigma\then\delta_0(x))\wedge(\neg\sigma\then x=0)
 \end{align*} have the properties we want.
\end{proof}

\begin{cor}\label{cor:ncl}\partenum \leavevmode
\begin{enumerate}[leftmargin=7mm]
\item There is an inductive formula that is not a cut.\label{part:ncl/ncut}
\item There is a cut that is not an a-cut.\label{part:ncl/n+}
\item There is an a-cut that is not an am-cut.\label{part:ncl/+nx}
\end{enumerate}
\end{cor}

\begin{proof}
Let $\phi(x)$ and $\delta(x)$ be as given by Lemma~\ref{lem:+nx}.
\begin{enumerate}\partenum
\item Take
\begin{math}
 \phi(x)\vee
 \ex c{(\delta(c)\wedge x\geq c^2)}
\end{math}.

\item Take
 \begin{math}
  \ex{c,z}{(\phi(z)\wedge\delta(c)\wedge x\leq c^2+z)}
 \end{math}.

\item Take $\phi(x)$. \qedhere
\end{enumerate}
\end{proof}

Solovay's method of shortening cuts~\cite[Theorem~III.3.5]{book:hajek+pudlak}
 provides a general way of producing cuts with various closure properties.
However, a technique for producing cuts with specific non-closure properties
 in large quantities does not seem available at present.
In particular, it is not clear whether
 there is an inverse of Solovay's method.
\begin{qu}\label{qn:lengthen-am}
Let $\phi(x)$ be a cut.
Assuming
 \begin{math}
  \PAminus\nproves\fa x{\phi(x)}
 \end{math},
 can one always find a cut $\psi(x)$
  such that $\PAminus$ proves
  \begin{equation*}
   \fa x{\bigl(\phi(x)\then\psi(x)\bigr)}
   \qand\fa x{\phi(x)}\nsc\fa x{\psi(x)}
  \end{equation*} and $\psi(x)$ is \emph{not} an am-cut?
\end{qu}

\section{Comparing solutions}\label{sec:comparing_solutions}

Given a $\PA$-theorem $\sigma$, our formulation of {\ITP} considers the set
of solutions, the search space, to be all inductive $\varphi(x)$
such that $\PAminus \proves \fa x{\phi(x)} \then \sigma$. In automated theorem proving,
restrictions of the search space (usually such that completeness is retained)
play an important role. In this short section, based on the results of Section~\ref{sec:non-closure_of_cuts},
we make a comment on comparisons of two solutions $\varphi_0(x)$ and $\varphi_1(x)$.

There are (at least) two ways of
 comparing two induction axioms $\fa x{\phi_0(x)}$ and $\fa x{\phi_1(x)}$,
  where $\phi_0(x)$ and $\phi_1(x)$ are inductive formulas.
The first way is to say
 $\fa x{\phi_0(x)}$ is stronger than $\fa x{\phi_1(x)}$
 when $\PAminus\proves\fa x{\phi_0(x)}\then\fa x{\phi_1(x)}$.
The second way is to say
 $\fa x{\phi_0(x)}$ is stronger than $\fa x{\phi_1(x)}$
 when $\PAminus\proves\fa x{(\phi_0(x)\then\phi_1(x))}$.
Clearly, if an induction axiom is stronger than another one
  in the second sense,
 then it is also stronger in the first sense.
The converse, however, is not true.

\begin{prop}\label{prop:weaker-indn}
There are cuts $\phi_0(x),\phi_1(x)$ such that $\PAminus$ proves
\begin{equation*}
 \fa x{\phi_0(x)}\nsc\fa x{\phi_1(x)}
 \qand\fa x{(\phi_0(x)\then\phi_1(x))}
\end{equation*}
but
\begin{math}
 \PAminus\nproves\fa x{(\phi_1(x)\then\phi_0(x))}
\end{math}.
\end{prop}

\begin{proof}
Let $\phi(x)$ be as given by Lemma~\ref{lem:+nx}.
We show that $\phi_0(x)=\phi(x\times x)$ and $\phi_1(x)=\phi(x)$
 have the properties we want.

As $\phi(x)$~is an a-cut,
 it is not hard to check that $\phi_0(x)$~is a cut.
Next, notice
 $\PAminus\proves\fa x{\phi(x)}\then\fa x{\phi(x\times x)}$ trivially.
Thus $\PAminus\proves\fa x{\phi_1(x)}\then\fa x{\phi_0(x)}$.
Notice also that \begin{math}
 \PAminus\proves\fa x{\bigl(\phi(x\times x)\then\phi(x)\bigr)}
\end{math} because $\phi(x)$ is a cut.
So \begin{math}
 \PAminus\proves\fa x{(\phi_0(x)\then\phi_1(x))}
\end{math}.
Finally, we know $\PAminus\nproves\fa x{(\phi_1(x)\then\phi_0(x))}$
 since $\phi(x)$ is not an am-cut.
\end{proof}

We can view Proposition~\ref{prop:weaker-indn} from another angle.
One may expect that the weakest induction axiom
 which can prove an induction axiom~$\ind\phi_0$ is $\ind\phi_0$ itself.
The proposition above says that this is not the case
 \emph{if} we adopt the second way of comparing induction axioms:
  $\ind\phi_0$~may be provable from another induction axiom~$\ind\phi_1$
   where $\phi_1(x)$ is a strictly weaker than~$\phi_0(x)$
    over~$\PAminus$ as formulas.
One may question whether there is always a weakest induction axiom
 for proving a given theorem of~$\PA$ in this sense.
We do not know the answer.

\begin{qu}\label{qn:lengthen-cut}
Let $\phi_0(x)$ be a cut.
Assuming
 \begin{math}
  \PAminus\nproves\fa x{\phi_0(x)}
 \end{math},
 can one always find a cut $\phi_1(x)$
  such that $\PAminus$ proves
  \begin{equation*}
   \fa x{\phi_0(x)}\nsc\fa x{\phi_1(x)}
   \qand\fa x{(\phi_0(x)\then\phi_1(x))}
  \end{equation*}
  but
  \begin{math}
   \PAminus\nproves\fa x{(\phi_1(x)\then\phi_0(x))}
  \end{math}?
\end{qu}

As one can show by imitating our proof of Proposition~\ref{prop:weaker-indn},
 a positive answer to Question~\ref{qn:lengthen-am}
  implies a positive answer to Question~\ref{qn:lengthen-cut}.

\section{Comparing notions of inductiveness}\label{sec:comparing_notions}

The question of how
an inductive theorem prover should choose the induction rule
to be applied to its current goal has received a great deal of attention in the literature on inductive theorem proving.
Pioneering work on this question has been done in the context of the ACL2~system and its predecessor NQTHM through
the introduction of the recursion analysis technique~\cite{Boyer79Computational}.
This technique, which suggests an induction rule based on the type of recursion present in the goal, is specific to goals
involving primitive recursive functions; see also Bundy et al.~\cite{Bundy89Rational} and Stevens~\cite{Stevens88Rational}.
However, the choice of the induction rule also plays an important role in other contexts; see the discussion on
``induction revision'' in the book~\cite{Bundy05Rippling} by Bundy et al., for example.

The interest in this question comes from the tension between the choice of the induction rule and the choice of the induction formula when proving a goal: the more
flexibility we have in choosing the induction rule, the less flexibility we need in choosing the induction formula. In
the very extreme case, one can fix an induction formula, e.g., the goal, and search for an induction rule with respect to which
this formula is inductive. Thus one can dispose of the difficult task of finding a non-analytic induction formula
and simply search for a suitable induction rule. In this section we will carry out a comparison of the formulations
of induction introduced in Section~\ref{sec:forms_of_induction} from this point of view: we fix a
formula $\varphi(x)$ and ask with respect to which induction schemes it is inductive.
Let $X$ be any notion of inductiveness defined in Section~\ref{sec:forms_of_induction}. Then we can state
the computational problem
\cprob{$\ITP^X$}
{A sentence $\sigma$ provable in $\PA$}
{An $X$-inductive formula $\varphi(x)$ s.t.\ $\PAminus \proves \fa x{\phi(x)} \then \sigma$}
By Proposition~\ref{prop:altPA} all these problems define total relations. By
Theorem~\ref{thm:IOpen}, the quantifier-free versions of $\ITP^X$ are equivalent
for all $X$ except possibly polynomial induction.

In this section, we investigate implications between these notions of inductiveness.
As we will see, some notions turn out to be incomparable.
For instance, not all $2$-step inductive formulas are $3$-step inductive, and
 not all $3$-step inductive formulas are $2$-step inductive.
The following proposition lists all the implications we can establish.
\begin{prop}\label{prop:indrel}
Let $m,n\in\IN$.
\begin{enumerate}\partenum
\item The following are equivalent for a formula:
 it is inductive;
 it is $1$-step inductive; and
 it is $1$-inductive.\label{part:indrel/1}
\item If $m+1$ divides $n+1$,
 then all $(m+1)$-step inductive formulas are $(n+1)$-step inductive.
\item If $m\leq n$,
 then all $(m+1)$-inductive formulas are $(n+1)$-inductive.\label{part:indrel/ind}
\item All $(n+1)$-step inductive formulas are $(n+1)$-inductive.\label{part:indrel/step>n}
\item\label{prop:indrel:n+1_implies_<} If a formula is $(n+1)$-inductive
   for some $n\in\IN$,
  then it is $<$-inductive.
\item A formula~$\phi(x)$ is $<$-inductive
  if and only if $\fale{x'}x{\phi(x')}$ is inductive.\label{part:indrel/ind'=wind}
\end{enumerate}
\end{prop}

\begin{proof}
Parts \partref{part:indrel/1}--\partref{part:indrel/step>n}
 are easy exercises.

\begin{enumerate}\partenum\addtocounter{enumi}{4}
\item Suppose $\phi(x)$ is $(n+1)$-inductive.
Work over~$\PAminus$.
Let $y_0$ be such that $\falt x{y_0}{\phi(x)}$.
If $y_0\leq n$, then $\phi(y_0)$ holds
 by Lemma~\ref{lem:PA-}\partref{part:PA-/ee}
  because $\phi(x)$ is $(n+1)$-inductive.
So suppose $y_0>n$.
Use axiom~\axref{ax:PA-/<} to find $x_0$
 such that $y_0=x_0+n+1$.
Then $\bigwwedge_{k<n+1}\phi(x_0+k)$
 because $x_0+k<y_0$ for each $k<n+1$ by~\axref{ax:PA-/<} again.
As $\phi(x)$ is $(n+1)$-inductive and $y_0=x_0+n+1$,
 this implies $\phi(y_0)$.

\item
First, assume $\fale{x'}x{\phi(x')}$ is inductive.
Work over~$\PAminus$.
Let $y_0$ be such that $\falt x{y_0}{\phi(x)}$.
If $y_0=0$, then $\phi(y_0)$ by our assumption.
So suppose $y_0\not=0$.
Apply Lemma~\ref{lem:PA-}\partref{part:PA-/pred}
 to find $y_0'$ such that $y_0=y_0'+1$.
Notice $\fale x{y_0'}{\phi(x)}$ by Lemma~\ref{lem:PA-}\partref{part:PA-/x<x+1}.
Our assumption then implies $\fale x{y_0}{\phi(x)}$,
 which, in particular, tells us that $\phi(y_0)$.

Conversely, suppose $\phi(x)$ is $<$-inductive.
Work over $\PAminus$ again.
With the help of~\axref{ax:PA-/nneg},
                 \axref{ax:PA-/<trans}~and~\axref{ax:PA-/<irref},
 we have $\phi(0)$ trivially by $<$-inductiveness.
So $\fale{x'}0{\phi(x')}$.
Let $x_0$ be such that $\fale{x'}{x_0}{\phi(x')}$.
Then Lemma~\ref{lem:PA-}\partref{part:PA-/discrete}
                    and~\partref{part:PA-/<+cancel}
  imply $\falt{x'}{x_0+1}{\phi(x')}$ and
 so $\phi(x_0+1)$ by $<$-inductiveness.
This shows $\fale{x'}{x_0+1}{\phi(x')}$, as required. \qedhere
\end{enumerate}
\end{proof}

We now proceed to show non-implications (based on the non-closure properties
 of cuts established in Section~\ref{sec:non-closure_of_cuts}).
In particular, we show that beyond Proposition~\ref{prop:indrel} there is no further implication
 between ``$<$-inductive'', the ``$(n+1)$-step inductive''s, and
  the ``$(n+1)$-inductive''s.
\begin{prop}\label{prop:indrel-strict}
Let $m,n\in\IN$.
\begin{enumerate}\partenum
\item If $(m+1)$ does not divide $(n+1)$,
 then there is an $(m+1)$-step inductive formula
  that is not $(n+1)$-step inductive.
\item If $m>n$, then\label{part:indrel-strict/ind}
 there is an $(m+1)$-inductive formula that is not $(n+1)$-inductive.
\item If $n\geq1$, then
 there is an $(n+1)$-inductive formula that is not $(n+1)$-step inductive.
\item\label{prop:indrel-strict:n+1_implies_<} There is a $<$-inductive formula
  that is not $(n+1)$-inductive.
\end{enumerate}
\end{prop}

\begin{proof}
Let $\phi(x)$ and~$\delta(x)$ be $\Lor$~formulas
 given by Lemma~\ref{lem:+nx}.
Recall $\phi(x)$ is not an am-cut.
Take $c\in M\models\PAminus+\neg\fa x{(\phi(x)\then\phi(x^2))}+\delta(c)$.
\begin{enumerate}\partenum
\item
Assume all $(m+1)$-step inductive formulas are $(n+1)$-step inductive.
Consider the formula \begin{equation*}
  \chi(x) \qequals \phi(x)\vee\ex y{(x=(m+1)y)}.
 \end{equation*}

We first show that $\chi(x)$ is $(m+1)$-step inductive.
By our assumption on $m$ and~$n$,
 this will imply $\chi(x)$ is $(n+1)$-step inductive.
Work over $\PAminus$.
We know $\bigwwedge_{k<m+1}\chi(k)$ because $\phi(x)$ is inductive.
Take $x_0$ satisfying $\chi(x_0)$.
If $y_0$ is such that $x_0=(m+1)y_0$,
 then $x_0+m+1=(m+1)(y_0+1)$ by~\axref{ax:PA-/x1} and~\axref{ax:PA-/distrib},
  and thus $\chi(x_0+m+1)$.
So suppose $\neg\ex y{(x_0=(m+1)y)}$.
Then $\phi(x_0)$ holds by the definition of the formula~$\chi(x)$.
As $\phi(x)$ is an inductive formula,
 this implies $\phi(x_0+m+1)$ and hence $\chi(x_0+m+1)$.

Now look at our model~$M$.
Let $a=(m+1)(n+1)c^2$, so that $M\models\chi(a)$.
Notice $M\models\neg\phi(c^2)$
 by condition~\condref{cond:+nx/c2>phi} in Lemma~\ref{lem:+nx}.
So $M\models\neg\phi(a+n+1)$
 because $\phi(x)$ is a cut and
  $a+n+1=c^2+c^2(mn+m+n)+n+1>c^2$ by axiom~\axref{ax:PA-/<}.
However, the previous paragraph tells us that $M\models\chi(a+n+1)$.
So $M\models\ex y{(a+n+1=(m+1)y)}$ by the definition of~$\chi(x)$.

Find $b\in M$ such that $a+n+1=(m+1)b$.
By axiom~\axref{ax:PA-/<}, we know
 \begin{math}
  (m+1)b=a+n+1>a=(m+1)(n+1)c^2
 \end{math}.
So Lemma~\ref{lem:PA-}\partref{part:PA-/<xcancel}
 implies $b>(n+1)c^2$.
Apply axiom~\axref{ax:PA-/<} to find $z\in M$
 such that $(n+1)c^2+z+1=b$.
Then
 \begin{equation*}
  a+(m+1)(z+1)=(m+1)(n+1)c^2+(m+1)(z+1)=(m+1)b=a+n+1
 \end{equation*} and
 thus $(m+1)(z+1)=n+1$
  by axioms~\axref{ax:PA-/<linear} and~\axref{ax:PA-/<+}.
Now $n+1=z+m(z+1)+1$.
Hence $n+1>z$ by axiom~\axref{ax:PA-/<}.
Lemma~\ref{lem:PA-}\partref{part:PA-/ee} then tells us $z\in\IN$.
We can thus conclude that $(m+1)$ divides $(n+1)$,
 which is what we want.

\item
We show that the formula \begin{equation*}
  \rho(x) \qequals \phi(x)\vee\fa c{\Bigl(
   \delta(c)\then \bigvvee_{k<n+1}x=c^2+k
  \Bigr)}
 \end{equation*} is an example we want.

Suppose $m>n$.
We show that $\rho(x)$ is $m$-inductive.
Work over $\PAminus$.
If $\fa x{(\phi(x)\then\phi(x^2))}$,
 then $\fa x{(\rho(x)\nsc\phi(x))}$
   by condition~\condref{cond:+nx/c-in-phi} in Lemma~\ref{lem:+nx}, and
  so we are done by Proposition~\ref{prop:indrel}\partref{part:indrel/1}
                                             and~\partref{part:indrel/ind}.
So suppose $\neg\fa x{(\phi(x)\then\phi(x^2))}$.
Since $\phi(x)$ is inductive, we know $\bigwwedge_{k<m+1}\rho(k)$.
Let $c,x_0$ be such that $\delta(c)\wedge\bigwwedge_{k<m+1}\rho(x_0+k)$.
Take $k<n+1$.
We know $c^2+k+m\geq c^2+m>c^2+n$
  by axiom~\axref{ax:PA-/<+}, and
 $\neg\phi(c^2+k+m)$ since $\phi(x)$ is a cut and $\neg\phi(c^2)$.
These together imply $\neg\rho(c^2+k+m)$.
So $x_0\not=c^2+k$ because $\rho(x_0+m)$ holds.
As the choice of $k<n+1$ is arbitrary,
 we know $\phi(x_0)$ by the definition of~$\rho(x)$.
It follows that $\phi(x_0+m+1)$ because $\phi(x)$ is inductive.
Hence $\rho(x_0+m+1)$.

We now show that $\rho(x)$ is not $(n+1)$-inductive using our model~$M$.
The definition of~$\rho(x)$ tells us
 $M\models\bigwwedge_{k<n+1}\rho(c^2+k)$.
On the one hand, notice $M\models\neg\phi(c^2)$
 by condition~\condref{cond:+nx/c2>phi} in Lemma~\ref{lem:+nx}.
Also $c^2+n+1>c^2$ by axiom~\axref{ax:PA-/<}.
Hence $\neg\phi(c^2+n+1)$ since $\phi(x)$ is a cut.
On the other hand, notice $c^2+n+1>c^2+k$ for any $k<n+1$
 by axiom~\axref{ax:PA-/<+}.
These together say that $M\models\neg\rho(c^2+n+1)$.

\item
We claim that the formula \begin{equation*}
  \rho_0(x) \qequals \phi(x)\vee\fa c{(\delta(c)\then x=c^2)}
 \end{equation*}
 has the desired properties.

To show that $\rho_0(x)$ is $(n+1)$-inductive, let us work over $\PAminus$.
As in the previous part,
 we assume $\neg\fa x{(\phi(x)\then\phi(x^2))}$.
Since $\phi(x)$ is inductive, we know $\bigwwedge_{k<n+1}\rho_0(k)$.
Let $c,x_0$ be such that $\delta(c)\wedge\bigwwedge_{k<n+1}\rho_0(x_0+k)$.
As $n\geq1$, we know $\rho_0(x_0)\wedge\rho_0(x_0+1)$ in particular.
Notice $\neg\phi(c^2+1)$
 because $\phi(x)$ is a cut and $\neg\phi(c^2)$
  by condition~\condref{cond:+nx/c2>phi} in Lemma~\ref{lem:+nx}.
Also $c^2+1\not=c^2$ by Lemma~\ref{lem:PA-}\partref{part:PA-/x<x+1}.
Thus $\neg\rho_0(c^2+1)$ holds, from which one deduces $x_0\not=c^2$.
This implies $\phi(x_0)$ by the definition of~$\rho_0(x)$.
Since $\phi(x)$ is inductive, we conclude that $\phi(x_0+n+1)$.

To show that $\rho_0(x)$ is not $(n+1)$-step inductive,
 we use~$M$ as in the proof of~\partref{part:indrel-strict/ind}.
Notice $M\models\rho_0(c^2)$ by the definition of~$\rho_0(x)$.
On the one hand, we know $M\models\neg\phi(c^2)$
 by condition~\condref{cond:+nx/c2>phi} in Lemma~\ref{lem:+nx}.
Also $c^2+n+1>c^2$ by axiom~\axref{ax:PA-/<}.
Hence $\neg\phi(c^2+n+1)$ since $\phi(x)$ is a cut.
On the other hand, we know $c^2+n+1>c^2$ by~\axref{ax:PA-/<}.
These together say that $M\models\neg\rho_0(c^2+n+1)$.

\item
It suffices to prove that the formula~$\rho(x)$ defined
  in the proof of~\partref{part:indrel-strict/ind}
 is $<$-inductive.

Work over~$\PAminus$.
Let $c,y_0$ be such that $\delta(c)$ and $\falt x{y_0}{\rho(x)}$.
We will show $\rho(y_0)$.
If $y_0=0$, then $\phi(y_0)$ because $\phi(x)$ is inductive,
 and so $\rho(y_0)$ holds.
Therefore, assume $y_0\not=0$.
Apply Lemma~\ref{lem:PA-}\partref{part:PA-/pred}
 to find $x_0$ such that $y_0=x_0+1$.
Notice $y_0>x_0$ by Lemma~\ref{lem:PA-}\partref{part:PA-/x<x+1}.
So, by the choice of~$y_0$, we know $\rho(x_0)$.

Consider the case when $\neg\fa x{(\phi(x)\then\phi(x^2))}$.
Then $\neg\phi(c^2)$
 by condition~\condref{cond:+nx/c2>phi} in Lemma~\ref{lem:+nx}.
This implies $c^2\not=0$ because $\phi(x)$ is inductive.
Use Lemma~\ref{lem:PA-}\partref{part:PA-/pred}
 to find $w$ such that $c^2=w+1$.
Then $\neg\phi(w)$ since $\phi(x)$ is inductive and $\neg\phi(w+1)$ holds.
Also, Lemma~\ref{lem:PA-}\partref{part:PA-/x<x+1} implies
 $w<c^2\leq c^2+k$ for each $k<n+1$.
We thus know $\neg\rho(w)$ by the definition of~$\rho(x)$.
Our assumption on~$y_0$ then implies $x_0<x_0+1=y_0\leq w<w+1=c^2$.
Thus $x_0\not=c^2+k$ for any $k<n+1$
  by axiom~\axref{ax:PA-/<trans}, and
 so $\phi(x_0)$ must hold.

If $\fa x{(\phi(x)\then\phi(x^2))}$,
 then $\bigwwedge_{k<n+1}\phi(c^2+k)$
  by conditions~\condref{cond:+nx/cut} and~\condref{cond:+nx/c-in-phi}
   in Lemma~\ref{lem:+nx}.
So $\phi(x_0)$ in either case because $\rho(x_0)$~holds.
Since $\phi(x)$ is inductive,
 we deduce $\phi(x_0+1)$ and hence $\rho(y_0)$, as required. \qedhere
\end{enumerate}
\end{proof}

Buss's notion of polynomial induction
 does not fit well into the picture.
\begin{prop}\label{prop:p-ind-compare}
\begin{enumerate}\partenum
\item There is a formula that is $(n+1)$-step inductive for every $n\in\IN$
  but is not \p-inductive.
\item For every $n\in\IN$,
 there is a \p-inductive formula that is not $(n+1)$-inductive.
\end{enumerate}
\end{prop}

\begin{proof}
\begin{enumerate}\partenum
\item
By Corollary~\ref{cor:ncl}\partref{part:ncl/n+},
 there is a cut that is not an a-cut.
Such a cut is $(n+1)$-step inductive for every $n\in\IN$,
 but it is not \p-inductive.

\item
Let $\phi(x),\delta(x)$ be $\Lor$~formulas
 given by Lemma~\ref{lem:+nx}.
Define $\chi(x)$ to be the formula
 \begin{equation*}
  \phi(x)\vee\ex c{\ex s{\ex\ell{\left(\begin{aligned}
   &\delta(c)\wedge\bigvvee_{k<n+1}(s)_0=c^2+k\wedge(s)_\ell=x\\
   &\wedge\falt i\ell{\bigl((s)_{i+1}=2(s)_i\vee(s)_{i+1}=2(s)_i+1\bigr)}\\
   &\wedge\fale{i,j}\ell{\bigl(i<j\then(s)_i<(s)_j\bigr)}
  \end{aligned}\right)}}}.
 \end{equation*}
Here $(s)_i=x$ is the $\Lor$~formula expressing
 ``the $i$th element in the sequence coded by~$s$ is~$x$'' over~$\PAminus$
  given in Je\v r\'abek~\cite{art:jerabek/PAminus}.
It is then straightforward to see that $\chi(x)$ is \p-inductive.
However, the formula~$\chi(x)$ is not $(n+1)$-inductive
 because $\PAminus\proves\ex c{\bigl(
           \delta(c)\wedge\bigwwedge_{k<n+1}\chi(c^2+k)
          \bigr)}$ and
         $\PAminus\nproves\ex c{(\delta(c)\wedge\chi(c^2+n+1))}$. \qedhere
\end{enumerate}
\end{proof}

We do not know
 whether all \p-inductive formulas are $<$-inductive.
What we have is only a translation of this question
 into a model-theoretic language.
\begin{prop}\label{prop:p+nind'=defelt}
The following are equivalent.
\begin{enumerate}\tfaeenum
\item There is a \p-inductive formula that is not $<$-inductive.\label{cond:p+nind'=defelt/ceg}
\item In some model of~$\PAminus$,\label{cond:p+nind'=defelt/defelt}
  there is a parameter-free definable element
   that is neither even nor odd.
\end{enumerate}
\end{prop}

\begin{proof}
First, suppose \condref{cond:p+nind'=defelt/defelt}~fails.
Let $\phi(x)$ be any $\Lor$~formula that is not $<$-inductive.
Find \begin{math}
 c\in M\models\PAminus+
  \falt xc{\phi(x)}\wedge\neg\phi(c)
\end{math}.
If $c=0$, then $\PAminus\nproves\phi(0)$ and we are done.
So suppose $c\not=0$.
Notice $c=(\min x)(\neg\phi(x))$.
It is, therefore, a parameter-free definable element of~$M$.
The failure of~\condref{cond:p+nind'=defelt/defelt} then gives us $d\in M$
 such that $c=2d$ or $c=2d+1$.
By Lemma~\ref{lem:PA-}\partref{part:PA-/pred} and axiom~\axref{ax:PA-/<},
 we know $d<c$.
Hence $M\models\phi(d)$ by the minimality of~$c$.
This shows $\phi(x)$ is not \p-inductive.

Conversely, suppose~\condref{cond:p+nind'=defelt/defelt}~holds.
Let $c\in M\models\PAminus$
 in which $c$~is a non-even non-odd parameter-free definable element.
Find an $\Lor$~formula $\delta(x)$ that defines~$c$ in~$M$.
Define $\phi_0(x)$ to be
 \begin{equation*}
  \exi y{\delta(y)}\then\ex y{(\delta(y)\wedge x<y)},
 \end{equation*} and
 let $\phi(x)$ be
 \begin{equation*}
  \ex s{\ex\ell{\bigl(
   \phi_0((s)_0)\wedge(s)_\ell=x
   \wedge\falt i\ell{\bigl(
    (s)_{i+1}=2(s)_i\vee(s)_{i+1}=2(s)_i+1
   \bigr)}
  \bigr)}},
 \end{equation*}
 where $(s)_i=x$ is the $\Lor$~formula expressing
  ``the $i$th element in the sequence coded by~$s$ is~$x$''
   in our proof of Proposition~\ref{prop:p-ind-compare}.
Then $\phi(x)$ is \p-inductive by construction.
Notice $M\models\falt xc{\phi_0(x)}$.
So $M\models\falt xc{\phi(x)}$
 because $\PAminus\proves\fa x{(\phi_0(x)\then\phi(x))}$.
However, every element $x\in M$ satisfying~$\phi(x)$ must
 lie strictly below~$c$
 or be either even or odd.
Hence $M\models\neg\phi(c)$ by our choice of~$c$.
We thus know $\phi(x)$ is not $<$-inductive.
\end{proof}

\begin{rem}\label{rmk:bases}
All results in this paper about base-2 polynomial induction
 generalize to other bases.
Let us say an $\Lor$~formula~$\phi(x)$ is \defm{$(n+2)$-\p-inductive},
  where $n\in\IN$,
 if \begin{equation*}
  \PAminus\proves\phi(0)\wedge\fa x{\Bigl(
   \phi(x)\then\bigwwedge_{k<n+2}\phi\bigl((n+2)x+k\bigr)
  \Bigr)}.
 \end{equation*}
In this terminology,
 the \p-inductive formulas are precisely the $2$-\p-inductive formulas.
Let $m,n\in\IN$.
Then $n+2$ being a power of $m+2$ is a necessary and sufficient condition
 for every $(m+2)$-\p-inductive formula to be $(n+2)$-\p-inductive.
We omit the proof here.
\end{rem}

\section{Removing the order from the language}\label{sec:order}

Another aspect of our formulation of {\ITP} is that we search for an inductive formula in the same
language as the input sentence. In this section we consider
a restriction of the language of the inductive formula.
Since most of our arguments about~$\IOpen$
 depend on the presence (or, in fact, the quantifier-free definability)
  of the order~$<$ in~$\Lor$,
 it is natural to ask whether $\IOpen$ really becomes strictly weaker
  when the induction scheme is restricted
   to formulas in which $<$~does not appear.
It turns out that
 the induction scheme for quantifier-free formulas in this restricted language
  is already provable in the base theory~$\PAminus$.
According to Shepherdson~\cite[page~27]{art:sheph/rule}, this fact was first observed by M.~D. Gladstone.
A stronger result was established by Shepherdson~\cite[Theorem~2]{art:sheph/rule},
 who proved the same scheme from a weak subtheory of~$\PAminus$.
If one starts from~$\PAminus$,
 then one can actually prove the least number principle
   (also known as the scheme of strong induction,
    cf.~Proposition~\ref{prop:altPA}\itemref{item:altPA/'} and
        Theorem~\ref{thm:IOpen}\itemref{item:IOpen/'})
  for quantifier-free formulas in the restricted language.
This was pointed out by Richard Kaye in a conversation in April~2016.
With Kaye's permission, we include his proof here.
Many ideas in the proof were already present
 in Shepherdson's contribution~\cite{art:sheph/model}
  to Theorem~\ref{thm:IOpen}.

\begin{defi}
Denote the language $\{0,1,{+},{\times}\}$ for rings by \defm{$\Lr$}.
Define
 \begin{equation*}
  \IOpen(\Lr)=\PAminus+\{\ind_x\theta:
   \text{$\theta(x,\bar z)$ is a quantifier-free $\Lr$~formula}
  \}.
 \end{equation*}
Let \defm{$\LOpen(\Lr)$} be the theory axiomatized by~$\PAminus$ and
 the scheme
 \begin{equation*}
   \fa{\bar z}{\bigl(
    \fa y{\bigl(\falt xy{\theta(x,\bar z)\then\theta(y,\bar z)}\bigr)}
    \then\fa x{\theta(x,\bar z)}
   \bigr)},
  \end{equation*} where $\theta$ ranges over quantifier-free $\Lr$~formulas.
\end{defi}

The following elementary algebraic fact plays a key role in Kaye's proof.
\begin{lem}\label{lem:fin0}
Let $F$ be a field.
For every polynomial~$p(x,\bar a)$ in a single variable~$x$
  with coefficients $\bar a\in F$,
 the set
 \begin{math}
  \{x\in F:p(x,\bar a)=0\}
 \end{math}
 is either finite or equal to~$F$. \qed
\end{lem}

\begin{thm}[Kaye]\label{thm:PAminus-IOpen_r}
$\PAminus\proves\LOpen(\Lr)$.
\end{thm}

\begin{proof}
Let $M\models\PAminus$.
We will show $M\models\LOpen(\Lr)$.
Fix $\bar a\in M$.
Consider a quantifier-free $\Lr$~formula $\theta(x,\bar z)$.
Put $\neg\theta(x,\bar z)$ in disjunctive normal form
 \begin{equation*}
  \bigvvee_{i\leq m} \bigwwedge_{j\leq n} \bigl(
   p_{ij}(x,\bar z)=0\wedge q_{ij}(x,\bar z)\not=0
  \bigr),
 \end{equation*}
 where the $p_{ij}$'s and the $q_{ij}$'s are polynomials over~$\ZZ$.
It suffices to show that for each $i\leq m$,
 if some $x\in M$ satisfies
 \begin{math}
  \bigwwedge_{j\leq n} \bigl(
   p_{ij}(x,\bar a)=0\wedge q_{ij}(x,\bar a)\not=0
  \bigr)
 \end{math}, then there is a least such~$x$.
So for notational convenience,
 let us assume $m=0$, so that $\neg\theta(x,\bar z)$ becomes
  \begin{equation*}
   \bigwwedge_{j\leq n} \bigl(
    p_{0j}(x,\bar z)=0\wedge q_{0j}(x,\bar z)\not=0
   \bigr).
  \end{equation*}
Since $M$~is the non-negative part of a discretely ordered ring,
 we can further simplify the formula~$\neg\theta(x,\bar z)$ to
 \begin{equation*}
  p(x,\bar z)=0\wedge q(x,\bar z)\not=0
 \end{equation*}
 by setting $p(x,\bar z)=\sum_{j\leq n} (p_{0j}(x,\bar z))^2$ and
            $q(x,\bar z)=\prod_{j\leq n} q_{0j}(x,\bar z)$.

Assume
 \begin{math}
  M\models\fa y{\bigl(
   \falt xy{\theta(x,\bar a)\then\theta(y,\bar a)}
  \bigr)}
 \end{math}.
We will construct a strictly decreasing sequence $(c_\ell)_{\ell\in\IN}$
 of elements of~$M$ by recursion as follows.
If $M\models\fa x{\theta(x,\bar a)}$, then we are done already.
So suppose not, and take an arbitrary $c_0\in M\models\neg\theta(c_0,\bar a)$.
Next, suppose $c_\ell$ is found such that $M\models\neg\theta(c_\ell,\bar a)$.
Using our assumption, pick any $c_{\ell+1}<c_\ell$ in~$M$
 with $M\models\neg\theta(c_{\ell+1},\bar a)$ and carry on.

The result of this construction is an infinite sequence $c_0>c_1>c_2>\cdots$
  of elements of~$M$
 such that $p(c_\ell,\bar a)=0$ and $q(c_\ell,\bar a)\not=0$ for all $\ell\in\IN$.
Lemma~\ref{lem:fin0} then implies $p(x,\bar a)$ is the zero polynomial.
Hence, actually
 $M\models\fa x{\bigl(\neg\theta(x,\bar a)\nsc q(x,\bar a)\not=0\bigr)}$.
By an external induction using our assumption in the previous paragraph,
 one sees that $q(k,\bar a)=0$ for all $k\in\IN$.
Applying Lemma~\ref{lem:fin0} again,
 we conclude $M\models\fa x{q(x,\bar a)=0}$, as required.
\end{proof}

The obvious argument~\cite[Lemma~I.2.8]{book:hajek+pudlak}
 shows $\LOpen(\Lr)\proves\IOpen(\Lr)$.
Therefore, Lemma~\ref{lem:Div} and Theorem~\ref{thm:PAminus-IOpen_r}
 tell us that $\IOpen(\Lr)$ is strictly weaker than~$\IOpen$.
In terms of computational problems we obtain that
\cprob{\ITPOpenLR}
{A sentence $\sigma$ provable in $\PA$}
{A quantifier-free $\Lr$ formula $\theta(x,\bar z)$ s.t.\ $\PAminus+\ind_x\theta \proves \sigma$}
is defined on strictly fewer inputs than
\cprob{\ITPOpen}
{A sentence $\sigma$ provable in $\PA$}
{A quantifier-free $\Lor$~formula $\theta(x,\bar z)$ s.t.\ $\PAminus+\ind_x\theta \proves \sigma$}

To recover the strength of~$\IOpen$,
 it suffices to have induction for a single inequality.

 \newpage
\begin{thm}[Shepherdson]\label{thm.IAtomic=IOpen}
\begin{math}
 \PAminus+\{\ind_x \theta:\text{$\theta(x,\bar z)$ is an atomic $\Lor$~formula}\}
 \proves \IOpen
\end{math}.
\end{thm}

\begin{proof}
See page~83f in Shepherdson~\cite{art:sheph/model}.
\end{proof}

\section{Conclusion}

What do we learn from these results for the automation of inductive theorem proving?

In comparison with successor induction and the uniform proof shape {\ITPU}, both the use of $<$-induction
and that of the non-uniform proof shape {\ITP} come at the price of introducing a universal quantifier to the goals.
However, both also come with a significant pay-off:
there are strictly more $<$-inductive formulas than $(n+1)$-inductive formulas for any $n\in\IN$.
Similarly,
the uniform proof shape {\ITPU} is not complete while
the non-uniform proof shape {\ITP} is complete (with respect to $\PA$).
Therefore, assuming the theorem proving environment is able to deal with universal quantifiers in the assumptions of a goal efficiently,
these results clearly indicate to prefer $<$-induction over $(n+1)$-induction for any $n\in\IN$ and the {\ITP} proof
shape over the {\ITPU} proof shape.

As Theorem~\ref{thm:Th(PA)=indn} shows, the non-uniform proof shape {\ITP} can be restricted to the non-uniform equivalence
proof shape {\ITPEq} while preserving completeness, thus strongly reducing the search space. Whether this can be exploited
for practical applications in inductive theorem proving is unclear to the authors. What would be needed is a practically reasonable procedure for generating universal formulas which are $\PAminus$-equivalent
to a given sentence. To the best of the authors' knowledge, no such procedure has been devised for use in inductive theorem proving yet.
On a more theoretical note, Theorem~\ref{thm:Th(PA)=indn} shows that the essential difficulty of inductive theorem proving
does {\em not} lie in finding an inductive formula which is {\em stronger} than the goal (in the sense that it implies the goal over $\PAminus$);
it is sufficient to find an inductive formula as strong as the goal. Instead, the essential difficulty lies in the non-analyticity
of the inductive formula as discussed in Section~\ref{sec.nec_non_an}.

The results of Section~\ref{sec:order} illustrate the importance of the choice of the language for
the inductive formulas: allowing a single predicate symbol defined using a single quantifier can increase
the strength of the theory considerably.

From a broader perspective, we believe that this paper demonstrates the potential of the use
of methods and results from mathematical logic, and in particular, from theories of arithmetic,
in inductive theorem proving. In our opinion, a good illustration
of this potential is provided by Proposition~\ref{prop:indrel-strict}. The construction of theorem-specific
induction rules has a long history in inductive theorem
proving~\cite{Boyer79Computational,Bundy89Rational,Stevens88Rational,Bundy05Rippling}.
This approach is typically justified
empirically by demonstrating, in the context of a certain algorithm or implementation, that the option to switch to another
induction rule allows to solve more goals than before (by that particular algorithm or by a certain implementation within
a certain timeout). Proposition~\ref{prop:indrel-strict} provides a much more solid foundation to this approach:
if there are strictly more formulas which are inductive in sense X than in sense Y
an algorithm for inductive theorem proving can get trapped in a situation where
the induction formula currently under consideration allows a proof using rule X but does not allow a proof using rule Y,
{\em regardless} of which algorithm or implementation is used, since there simply is no proof.

All of the results in this paper are restricted to arithmetical theories in the strict sense, i.e., theories
about the natural numbers. In inductive theorem proving, one typically works in a broader syntactic setting
of many-sorted first-order logic with inductive data types such as lists, trees, etc.
Since many results about theories of arithmetic depend on results about numbers, rings, etc., it is not straightforward
to extend them to this more general syntactic setting. Coding is not an option since it introduces an
amount of syntactic complication that is not realistic in inductive theorem proving.
However, for the continuation of the line of work in this paper we consider
such an extension of the model theory of arithmetic to be of high importance. One prototypical problem along this line
is to classify (some of) the benchmark examples from the TIP suite~\cite{Claessen15TIP} according
to the weakest induction schemes (e.g.,~atomic, quantifier-free) required to prove them.

We believe, as argued in Section~\ref{sec:model}, that the non-analyticity of induction axioms is a key aspect
of inductive theorem proving, so we have considered the set of inductive formulas as the search space
in our model. In the case of theorem proving in pure first-order logic, the search space is quite well understood theoretically
and consequently we have very powerful tools at our disposal for navigating in it such as (most general)
unification, subsumption, reduction orderings, (resolution) refinements, etc.
At present, we do not seem to have a similarly complete theoretical understanding of the set of inductive
formulas and how they relate to a given goal. We believe that further work in this direction is crucial
for advancing the state of the art in inductive theorem proving.

\section{Acknowledgements}
The authors are very grateful to Emil Je\v r\'abek for pointing out
 some results in the literature~\cite{inproc:weakness-sbdd-pind,art:sheph/model,art:sheph/rule}
  that answer several questions posed in an earlier version of this paper.
The authors would like to thank
 Zofia Adamowicz, Alan Bundy, Richard Kaye, and Leszek Ko\l odziejczyk
  for illuminating discussions around the topics of this paper.
The authors would also like to thank the anonymous reviewers
 whose insightful feedback has led to a considerable improvement
  of this paper.

\bibliographystyle{alpha}
\bibliography{induction}

\end{document}